\documentclass{article}
\usepackage{amsmath, amssymb, bm, amsthm}
\usepackage[nolist,nohyperlinks]{acronym}
\usepackage[round]{natbib}
\usepackage{hyperref}
\usepackage[margin=2cm]{geometry}
\usepackage{tikz}
\usepackage{algorithm, algorithmic}
\usepackage{booktabs}
\usepackage{multirow}
\usepackage{makecell}
\usetikzlibrary{arrows.meta, positioning, calc, fit, shapes.geometric}
\usepackage{caption}
\usepackage{tabularx}
\usepackage{graphicx}   
\usepackage{threeparttable}  
\usepackage{pifont}
\usepackage{authblk}

\tikzset{
  node distance=1.2cm and 1.8cm,
  every node/.style={draw, circle, minimum size=1cm, font=\small},
  every path/.style={>=Stealth,->,thick}
}

\newtheorem{proposition}{Proposition}

\def\E{\mathbb{E}}

\title{Robust Mendelian Randomization Estimation using Weighted Quantile Regression}
\author[1,2]{Julien St-Pierre}
\author[3]{Archer Y. Yang}
\author[1,4]{Mireille E. Schnitzer}
\author[1,2]{Marc-André Legault}

\affil[1]{Faculté de Pharmacie, Université de Montréal, Montréal, QC, Canada}
\affil[2]{Centre de Recherche Azrieli du CHU Sainte-Justine, Montréal, QC, Canada}
\affil[3]{Department of Mathematics and Statistics, McGill University, Montréal, QC, Canada}
\affil[4]{Département de médecine sociale et préventive, Université de Montréal, Montréal, QC, Canada}

\begin{document}

\begin{acronym}
    \acro{2SLS}{two stage least squares}
\end{acronym}

\begin{acronym}
    \acro{MR}{Mendelian randomization}
\end{acronym}

\begin{acronym}
    \acro{KZCS}{\cite{Kang2016}}
\end{acronym}

\begin{acronym}
    \acro{OLS}{ordinary least squares}
\end{acronym}

\begin{acronym}
    \acro{LIML}{limited information maximum likelihood}
\end{acronym}

\begin{acronym}
    \acro{IV}{instrumental variable}
\end{acronym}

\begin{acronym}
    \acro{GWAS}{genome-wide association studies}
\end{acronym}

\begin{acronym}
    \acro{IVW}{inverse-variance weighted}
\end{acronym}

\begin{acronym}
    \acro{MCP}{minimax concave penalty}
\end{acronym}

\begin{acronym}
    \acro{QR}{quantile regression}
\end{acronym}

\begin{acronym}
    \acro{WQR}{weighted quantile regression}
\end{acronym}

\begin{acronym}
    \acro{MLE}{maximum-likelihood estimator}
\end{acronym}

\begin{acronym}
    \acro{ALD}{asymmetric Laplace distribution}
\end{acronym}

\begin{acronym}
    \acro{InSIDE}{Instrument Strength Independent of Direct Effect}
\end{acronym}

\begin{acronym}
    \acro{MAF}{minor allele frequency}
\end{acronym}

\begin{acronym}
    \acro{LAD}{least-absolute deviation}
\end{acronym}

\begin{acronym}
    \acro{CI}{confidence interval}
\end{acronym}

\begin{acronym}
    \acro{RHR}{resting heart rate}
\end{acronym}

\begin{acronym}
    \acro{AF}{atrial fibrillation}
\end{acronym}

\begin{acronym}
    \acro{SD}{standard deviation}
\end{acronym}

\begin{acronym}
    \acro{RMSE}{root mean square error}
\end{acronym}

\begin{acronym}
    \acro{OR}{odds ratio}
\end{acronym}

\begin{acronym}
    \acro{ZEMPA}{zero modal pleiotropy assumption}
\end{acronym}

\maketitle

\begin{abstract}
In Mendelian randomization (MR) studies, genetic variants are used as instrumental variables (IVs) to investigate causal relationships between exposures and outcomes based on observational data. However, numerous
genetic studies have shown the pervasive pleiotropy of genetic variants, meaning that many, if not most, variants are associated with multiple traits, potentially violating the core assumptions of IV estimation. Uncorrelated pleiotropy occurs when genetic variants have a direct effect on the outcome that is not mediated by the exposure, while correlated pleiotropy occurs when genetic variants affect the exposure and outcome via shared heritable confounders. In this work, we propose a novel MR method, called MR-Quantile, based on weighted quantile regression (WQR) that is robust to both correlated and uncorrelated pleiotropy. We propose a procedure for selecting the optimal quantile of the ratio estimates through a likelihood-based formulation of WQR using the asymmetric Laplace distribution. Monte Carlo simulations demonstrate the empirical performance of the proposed method, especially in settings with many invalid IVs with weak pleiotropic effects. Finally, we apply our method to study the causal effect of resting heart rate on atrial fibrillation. Genetic variants associated with heart rate were identified in a genome-wide association study of 425,748 individuals from the VA Million Veteran Program, and used as instruments in a two-sample MR analysis with summary statistics from a genetic meta-analysis of 228,926 AF cases across eight studies.
\end{abstract}

\section{Introduction}

In \ac{MR} studies, genetic variants, typically single nucleotide polymorphisms (SNPs), are used as \acp{IV} to investigate causal relationships between exposures and outcomes based on observational data. Notably, when individual-level data are unavailable, MR allows one to perform causal inference using only summary-level associations of genetic variants with the exposure and the outcome, which can easily be obtained from publicly available \ac{GWAS}. The validity of \ac{MR} studies strongly depends on three core assumptions required for a valid \ac{IV}, under the assumed data-generating model in Figure \ref{fig:dag}a, which are:~(A1) The \ac{IV} is associated with the exposure $X$;~(A2) The \ac{IV} is independent of the outcome $Y$ conditional on $X$ and the unmeasured confounder $U$; and (A3) The \ac{IV} is independent of $U$. When all instruments are valid and independent, the \ac{IVW} estimator~\citep{Burgess2013} provides a consistent estimate of the causal effect by combining variant-specific estimates, similar to a meta-analysis. However, numerous genetic studies have shown the pervasive pleiotropy of genetic variants, meaning that many, if not most, variants are associated with multiple traits~\citep{Mackay2024, Qi2024, Watanabe2019}.~Uncorrelated pleiotropy occurs when genetic variants have a direct effect on the outcome that is not mediated by the exposure, violating assumption A2, while correlated pleiotropy occurs when genetic variants affect the exposure and outcome via shared heritable confounders, violating assumption A3.

~\citet{Bowden2016} proposed a median estimator which is robust to violation of A2 and A3, and is consistent for the causal effect as long as at least half the SNPs are valid instruments, because this guarantees that the median estimate comes from a valid instrument~\citep{Windmeijer2018}. In practice, this assumption, which can be referred to as the majority valid assumption, is difficult to verify. \ac{MR} methods that depend on a weaker assumption, the plurality-valid assumption, have been proposed. This assumption states that, in large samples, the estimates from all valid \acp{IV} converge to the true causal effect, whereas those from invalid \acp{IV} converge to different values. Consequently, the valid \acp{IV} constitute the largest group of SNPs sharing a common ratio estimate. \citet{Hartwig2017} proposed using the mode of the smoothed empirical density function of all ratios as the causal effect estimate. MR-Lasso~\citep{Slob2020} and cML-MA-BIC~\citet{Xue2021} estimate the subset of invalid instruments with pleiotropic effects using either a lasso penalty selected via a heterogeneity criterion, or a constrained maximum likelihood (cML) approach combined with model averaging (MA) using Bayesian information criterion (BIC) weights. Alternatively, MR-Mix~\citep{Qi2019}, MR-ContMix~\citep{Burgess2020}, and MR-CAUSE~\citep{Morrison2020} model the heterogeneous distribution of valid and invalid SNPs using mixture models. While conceptually appealing, mixture model approaches can be difficult to fit when the number of SNPs is small, are often computationally demanding, and typically require estimating a large number of parameters. In contrast, methods that rely on selecting or removing invalid instruments may lack robustness in the presence of a high proportion of weak invalid IVs, where small but nonzero pleiotropic effects are difficult to distinguish from noise.

In this work, we propose a novel \ac{MR} method based on \ac{WQR} that is robust to both violations of A2 and A3 and depends only on the plurality-valid assumption. Our method, called MR-Quantile, is a generalization of the weighted median estimator to cases where possibly more than 50\% of \acp{IV} are invalid. We propose a data-driven procedure for selecting the optimal quantile of the ratio estimates through a likelihood-based formulation of \ac{WQR} using the \ac{ALD}. Due to its heavy tails, the \ac{ALD} is robust to outlying ratio estimates arising from invalid instruments with pleiotropic effects. Moreover, the \ac{ALD} naturally links likelihood-based inference to \ac{WQR}. If the sample distribution of ratio estimates contains a dominant mode corresponding to valid instruments, then the ALD provides a working likelihood that targets the $\tau$th quantile as the mode of this distribution. In Section 2, we review the weighted median estimator and its link with \ac{WQR}, before introducing our model based on the \ac{ALD}. In Section 3, we present results from Monte Carlo simulations comparing finite sample performance of different MR estimators under various scenarios. In Section 4, we implement the proposed estimator to study the causal effect of \ac{RHR} on \ac{AF} using summary statistics from 425,748 European ancestry individuals from the VA Million Veteran Program~\citep{Verma2024} and from a European meta-analysis which included 228,926 \ac{AF} cases from eight studies~\citep{Yuan2025}.

\begin{figure}[tbp]
\centering

\begin{tikzpicture}[
    node distance=2.2cm,
    thick,
    >=Stealth,
    every node/.style={font=\large}
]

\node[font=\bfseries] (labelA) at (-3.5,1.8) {a};
\node (IV) at (-4,0) {IV};
\node (X)  [right=of IV] {$X$};
\node (Y)  [right=of X] {$Y$};
\node (U)  at ($(X)!0.5!(Y)+(0,1.3)$) {$U$}; 

\draw[->] (IV) -- node[above] {A1} (X);
\draw[->] (X) -- (Y);
\draw[->] (U) -- (X);
\draw[->] (U) -- (Y);

\draw[<->, dashed] (IV) to[bend left=30] node[pos=0.5, above] {A3} (U);
\draw[->, dashed] (IV) to[bend right=30] node[pos=0.5, below] {A2} (Y);

\node at (-1.4,-0.95) {$\times$}; 
\node at (-2.2,1.35) {$\times$}; 

\begin{scope}[xshift=4cm]

\node[font=\bfseries] (labelB) at (1.2,1.8) {b};

\node (g) at (1.0,0) {$Z_i$};
\node (X2) [right=of g] {$X$};
\node (Y2) [right=of X2] {$Y$};
\node (U2) at ($(X2)!0.5!(Y2)+(0,1.3)$) {$U$}; 

\draw[->] (X2) -- node[above] {$\theta_0$} (Y2);
\draw[->] (U2) -- node[pos=0.2, above, right = 0.2] {$\beta_{YU}$} (Y2);
\draw[->] (U2) -- node[pos=0.2, above, left = 0.2]{$\beta_{XU}$} (X2);
\draw[->] (g) to[bend left=30] node[above] {$\phi_i$} (U2);
\draw[->] (g) to[bend right=30] node[below] {$\alpha_i$} (Y2);
\draw[->] (g) -- node[above] {$\gamma_i$}(X2);

\end{scope}

\end{tikzpicture}

\caption{\textbf{Causal model with exposure X and outcome Y.} \textbf{(a)} Three assumptions required for valid IVs. \textbf{(b)} Unified IV framework for valid and invalid IVs~\citep{Xue2021}.}
\label{fig:dag}
\end{figure}
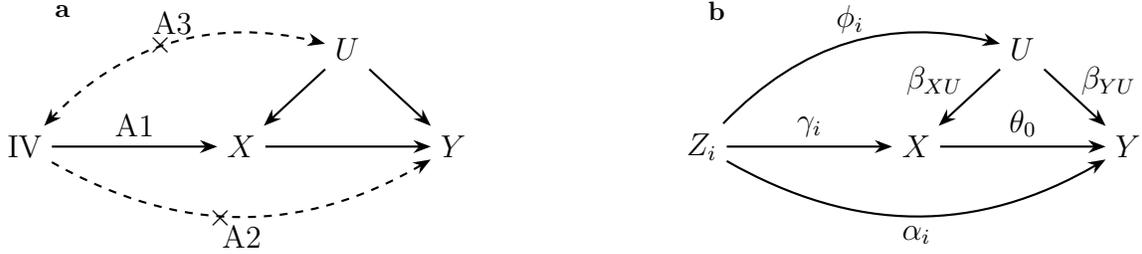

\section{Methods}
\subsection{Background}
A unified framework with both valid and invalid \acp{IV} as illustrated in Figure \ref{fig:dag}b has been proposed by \citet{Xue2021}, where $i=1,...,p$ indexes the number of independent genetic variants. Consider the structural equation model with a continuous outcome $Y$, a continuous exposure $X$, an unmeasured confounder $U$, and $p$ uncorrelated instruments $Z_1, ..., Z_p$ such that
\begin{subequations}\label{eq:1}
\begin{align}
Y_j &= \theta_0X_j + \sum_{i=1}^p \alpha_i Z_{ij} + \beta_{YU}U_j + \epsilon_{Y_j} , \label{eq:1a} \\
X_j &= \sum_{i=1}^p \gamma_i Z_{ij} + \beta_{XU}U_j + \epsilon_{X_j} \label{eq:1b}, \\
U_j &= \sum_{i=1}^p \phi_i Z_{ij} + \epsilon_{U_j} \label{eq:1c},
\end{align}
\end{subequations}
for samples $j=1,...,n$, and where $\epsilon_{Y_j}, \epsilon_{X_j}, \epsilon_{U_j}$ are mutually independent error terms that are uncorrelated with $Z_1,...,Z_p$, and $\phi_i$ and $\alpha_i$ denote the correlated (indirect) and uncorrelated (direct) pleiotropic effects, respectively. Further, we assume that ${\E[\epsilon_{Y_j}\mid Z_1,...,Z_p] = \E[\epsilon_{X_j}\mid Z_1,...,Z_p] = \E[\epsilon_{U_j}\mid Z_1,...,Z_p] = 0}$. From \eqref{eq:1a}, $\theta_0$ corresponds to the average change in the outcome for a unit change in the exposure, which we define as the average causal effect (ACE) of $X$ on $Y$. Assuming all variables are mean-centered without loss of generality, let $\beta_{Yi}=\E[Z_{ij}^2]^{-1}\E[Z_{ij}Y_j]$ and $\beta_{Xi}=\E[Z_{ij}^2]^{-1}\E[Z_{ij}X_j]$, such that the Wald ratio estimate $r_i$ for the $i^{th}$ SNP, $i=1,...,p$, is equal to
\begin{equation}
    r_i = \frac{\beta_{Yi}}{\beta_{Xi}} = \theta_0 + \frac{\alpha_i + \beta_{YU}\phi_i}{\beta_{Xi}}. \label{eq:ratios}
\end{equation} 
If $Z_i$ is a valid \ac{IV}, assumptions A2 and A3 imply $\alpha_i = 0$ and $\phi_i = 0$, so that $r_i = \theta_0$ and the ACE, $\theta_0$, can be identified from the population-level regression coefficients of $X$ and $Y$ on $Z_i$. In two-sample MR studies, the population parameters $\beta_{Xi}$ and $\beta_{Yi}$ are replaced by estimates $\hat\beta_{Xi}$ and $\hat\beta_{Yi}$ obtained from two independent \ac{GWAS} summary statistics for $X$ and $Y$ respectively. Assuming all instruments are valid, the \ac{IVW} estimator~\citep{Burgess2013} is consistent for $\theta_0$ and can be obtained by fitting the linear model
\begin{align*}
\hat\beta_{Yi} = \theta\hat\beta_{Xi} + \hat\sigma_{Y_i}\epsilon_i, \qquad \epsilon_i \overset{i.i.d.}{\sim} N(0, 1), \qquad i=1,...,p,
\end{align*}
where $\hat\sigma_{Yi}$ is the standard error of $\hat\beta_{Yi}$. Moreover, the IVW estimator can be expressed as a weighted mean of the measured ratios $\hat r_i$, that is,
\begin{equation}\label{eq:ivw}
    \hat\theta_{IVW} = \sum_{i=1}^p w_i \hat r_i,
\end{equation}
where $\hat r_i = \hat\beta_{Yi}/\hat\beta_{Xi}$ and $w_i = \hat\sigma_{Yi}^{-2}\hat\beta_{Xi}^2 / \sum_{k=1}^p\hat\sigma_{Yk}^{-2}\hat\beta_{Xk}^2$. Consequently, the contribution from ratio estimates with higher variance or weaker instrument strength is downweighted.

Binary disease outcomes are routinely investigated in \ac{MR} studies, and genetic associations between disease outcomes and SNPs are often reported using logistic regression models, such that $\hat\beta_{Yi} = \log \widehat{\text{OR}}(Y\mid Z_i)$. However, due to the non-collapsibility of the logistic regression model, the marginal causal log odds ratio is generally not identifiable under a logistic-linear structural equation model~\citep{Burgess2015, Didelez2007}. For a rare disease outcome $Y$, if we instead assume that the log probability of the outcome is linear in the exposure with no effect modification, then the structural outcome model can be written as
\begin{equation*}
\log \E[Y_j \mid X_j, Z_{1j}, ..., Z_{pj}, U_j] = \theta_0X_j + \sum_{i=1}^p\alpha_iZ_{ij} + \beta_{YU}U_j.
\end{equation*}
Under this model, $\theta_0$ corresponds to the causal log relative risk (RR) for a one unit increase of $X$ on $Y$~\citep{Didelez2010}. Let $\beta_{Yi}$ be the log-coefficient parameter from a loglinear regression of $Y$ on $Z_i$, denoted by $\log \text{RR}(Y\mid Z_i)$. The Wald ratio in \eqref{eq:ratios} is now equal to
\begin{equation}\label{eq:ratios2}
r_i = \frac{\log \text{RR}(Y \mid Z_i)}{\beta_{Xi}} = \theta_0 + \frac{\alpha_i + \beta_{YU}\phi_i}{\beta_{Xi}},
\end{equation}
where $\beta_{Xi}=\E[Z_{ij}^2]^{-1}\E[Z_{ij}X_j]$ again denotes the coefficient parameter from a linear regression of $X$ on $Z_i$. Thus, when the disease outcome prevalence is low in the population under study, we can approximate $\log \text{RR}(Y | Z_i)$ in \eqref{eq:ratios2} by $\log{\text{OR}}(Y | Z_i)$, such that the causal log relative risk $\theta_0$ is approximately identified using summary statistics $\hat\beta_{Yi}$ and $\hat\beta_{Xi}$ from valid IVs, respectively estimated via logistic and linear regression models from two independent GWAS. Finally, assuming all instruments are valid, a consistent estimate for the causal log RR $\theta_0$ can be obtained by the IVW estimator in \eqref{eq:ivw}.

\subsection{Weighted quantile regression}
Let $\hat r_i = \hat\beta_{Yi} / \hat\beta_{Xi}$ be the ratio of the SNP-outcome and SNP-exposure summary statistics for the $i^{th}$ SNP, and $\sigma_{Yi},\sigma_{Xi}$ the standard errors for $ \hat\beta_{Yi}$ and $ \hat\beta_{Xi}$ respectively. Under the assumption that more than 50\% of the SNPs come from valid instruments, the median of the ratio estimates is consistent for the causal effect $\theta_0$. In practice, when combining ratio estimates with varying precisions, estimates with higher variability are downweighted. A common approach is to use inverse variance weights given by $w_i = 1 / \textrm{Var}(\hat{r}_i) $. From the delta method for the variance of the ratio of two random variables~\citep{Burgess2015}, we have
\begin{align*}
\textrm{Var}(\hat{r}_i) \approx
\frac{\sigma^2_{Yi}}{\hat{\beta}_{Xi}^{\,2}}
\;+\;
\frac{\hat{\beta}_{Yi}^{\,2}\,\sigma^2_{Xi}}{\hat{\beta}_{Xi}^{\,4}}
\;-\;
2\,\frac{\hat{\beta}_{Yi}}{\hat{\beta}_{Xi}^{\,3}}\,
\operatorname{Cov}\!\big(\hat{\beta}_{Yi},\hat{\beta}_{Xi}\big),
\end{align*}
with the last term being equal to zero when the SNP-outcome and SNP-exposure summary statistics are obtained from two independent \ac{GWAS}. Hence, the weighted median estimator from \citet{Bowden2016} can be obtained as the minimizer of the weighted \ac{LAD} regression model
\begin{align*}
\underset{\theta}{\text{min }} \sum_{i=1}^p w_i|\hat r_i - \theta|,
\end{align*}
with inverse-variance weights $w_i$. The weighted \ac{LAD} regression can be generalized to obtain any weighted $\tau$\textsuperscript{th} sample quantile, $0 < \tau < 1$, as the solution to the \ac{WQR} minimization problem~\citep{Koenker2005}
\begin{align*}
\underset{\theta}{\text{min }} \tau\sum_{i:\hat r_i \ge \theta}^p w_i|\hat r_i - \theta| + (1 - \tau)\sum_{i:\hat r_i < \theta}^p w_i|\hat r_i - \theta|.
\end{align*}
More generally, let $\hat Q(\tau)$ be the weighted $\tau$-quantile of the \ac{MR} ratios $\hat r_1, ...,\hat r_p$, with strictly positive weights $w_1, ..., w_p$, then
\begin{align*}
\hat Q(\tau) = \underset{\theta}{\text{arg min }} \sum_{i=1}^pw_i\rho_{\tau}(\hat r_i - \theta),
\end{align*}
where  
\[
\rho_\tau(u) =
\begin{cases}
\tau u, & \text{if } u \ge 0, \\[6pt]
(\tau - 1)u, & \text{if } u < 0
\end{cases}
\]
is the check loss function. We show in Appendix A.1 that $\hat Q(\tau)$ is a consistent estimator for $\theta_0$ if $\tau$ satisfies
\begin{align*}
 q_{-} \le \ &\tau \le 1 -q_{+},
\end{align*}
where $q_{-}$ and $q_{+}$ are respectively the weighted proportion of invalid \acp{IV} with negative and positive pleiotropic effects. This assumption is known in the MR literature as the \ac{ZEMPA}, i.e., the weighted mode of the pleiotropic effects across all instruments is 0~\citep{Hartwig2017}. In practice, identifying the proportion of invalid instruments with pleiotropic effects is challenging, and the choice of the $\tau$-th quantile must rely on strong and unverifiable assumptions about the distribution of the invalid \acp{IV}. To address this, we propose an approach based on the connection between the \ac{ALD} and \ac{WQR}, which allows for data-driven estimation of $\tau$ and $\theta$ that remains agnostic to the true proportion of invalid variants in the observed sample, as long as \ac{ZEMPA} holds.

\subsection{Model setup}
We assume that the ratio $\hat r_i$ follows an \ac{ALD}~\citep{Yu2005} with density 
\begin{align}\label{eq:laplace}
f(\hat r_i; \theta, \tau, \lambda_i) = \lambda_i\tau(1-\tau)\textrm{exp}(-\lambda_i\rho_\tau(\hat r_i - \theta))
\end{align}
for $i=1,...,p$ and where $0<\tau<1$ is a skewness parameter, and $\lambda_i > 0$ is an inverse scale parameter. The \ac{ALD} is skewed to the left when $\tau > 0.5$, and skewed to the right when $\tau < 0.5$ as shown in Figure \ref{fig:ald-density}. When $\tau = 0.5$, the distribution is known as the Laplace double exponential distribution. Moreover, if $\hat r_i \sim ALD(\theta, \tau, \lambda_i)$, then $\Pr(\hat r_i \le \theta) = \tau$ and $\Pr(\hat r_i > \theta) = 1 - \tau$ such that $\theta$ is the $\tau^{th}$ quantile of the distribution. Finally, the expectation of $\hat r_i$ is given by $\E[\hat r_i] = \theta + \frac{1-2\tau}{\tau(1-\tau)\lambda_i}$, such that when $\tau=0.5$, the mode, median and mean are identical. From \eqref{eq:laplace}, the log-likelihood function for the model parameters $(\theta, \tau, \lambda_i)$ is given by
\begin{align}\label{eq:loglik1}
\ell(\theta,\tau,\lambda_i; \hat r_i) &= \sum_{i=1}^p \log \lambda_i + \sum_{i=1}^p \log \tau(1-\tau) - \sum_{i=1}^p \lambda_i\rho_\tau(\hat r_i - \theta). 
\end{align}
To account for the varying precision of the ratio estimates, we further assume $\lambda_i=w_i\lambda$ with $w_i = 1 / \sqrt{\textrm{Var}(\hat r_i)}$ and $\lambda>0$ a common inverse scale parameter, such that \eqref{eq:loglik1} becomes
\begin{align}\label{eq:loglik2}
\ell(\theta,\tau,\lambda) &= \sum_{i=1}^p \log w_i + \sum_{i=1}^p \log \lambda\tau(1-\tau) - \lambda\sum_{i=1}^p w_i\rho_\tau(\hat r_i - \theta). 
\end{align}
The \acp{MLE} for $\theta$, $\tau$, and $\lambda$ are obtained by maximizing \eqref{eq:loglik2} with respect to each parameter while holding the other two fixed, iteratively until convergence, from where
\begin{align}
\hat\theta &= \underset{\theta}{\arg\min}\sum_{i=1}^p w_i\rho_{\hat\tau}(\hat r_i - \theta), \label{eq:theta} \\
\hat\lambda &= \frac{p}{\sum_{i=1}^p w_i\rho_{\hat\tau}(\hat r_i - \hat\theta)}, \label{eq:lambda} \\
\hat\tau &= 0.5 - \frac{\hat a}{2(2p + \sqrt{\hat a^2 + 4p^2})}, \label{eq:tau}
\end{align}
with $\hat a = \hat\lambda\sum_{i=1}^p w_i(\hat r_i - \hat \theta)$. The derivation of the \ac{MLE} for $\tau$ is detailed in Appendix A.2, and the iterative procedure is presented in Algorithm 1.~Hence, $\hat{\theta}$ is the $\hat\tau^{th}$ weighted quantile of the sample $\hat r_1, ..., \hat r_p$ with respective weights given by the inverse of the standard errors of the ratio estimates. Of note, when $\tau=0.5$, the \ac{MLE} in \eqref{eq:theta} is slightly different from the weighted median estimator of~\citet{Bowden2016} which uses the inverse of the variance of the ratio estimates as weights. 

To obtain a standard error and \ac{CI} for $\hat\theta$, we use the same approach as~\citet{Bowden2016} who proposed a parametric bootstrap approach.~For $b=1, ..., B$, we generate a bootstrap sample of the SNP-outcome and SNP-exposure summary statistics
\begin{align*}
&\hat\beta_{Yi}^{(b)} \sim N(\hat\beta_{Yi}, \sigma^2_{Y_i}), \\
&\hat\beta_{Xi}^{(b)} \sim N(\hat\beta_{Xi}, \sigma^2_{X_i}),
\end{align*}
for each variant $i=1,...,p$. To properly account for the uncertainty in the estimation of $\tau$ and $\lambda$ when calculating the standard error for $\hat\theta$, we re-estimate the \acp{MLE} for $\theta, \lambda, \tau$ for each bootstrap sample, and take the standard deviation of the bootstrap estimates $\hat\theta^{(1)}, ..., \hat\theta^{(B)}$ as the standard error of $\hat\theta$.

\begin{figure}[htbp]
    \centering
    \includegraphics[width=0.8\textwidth]{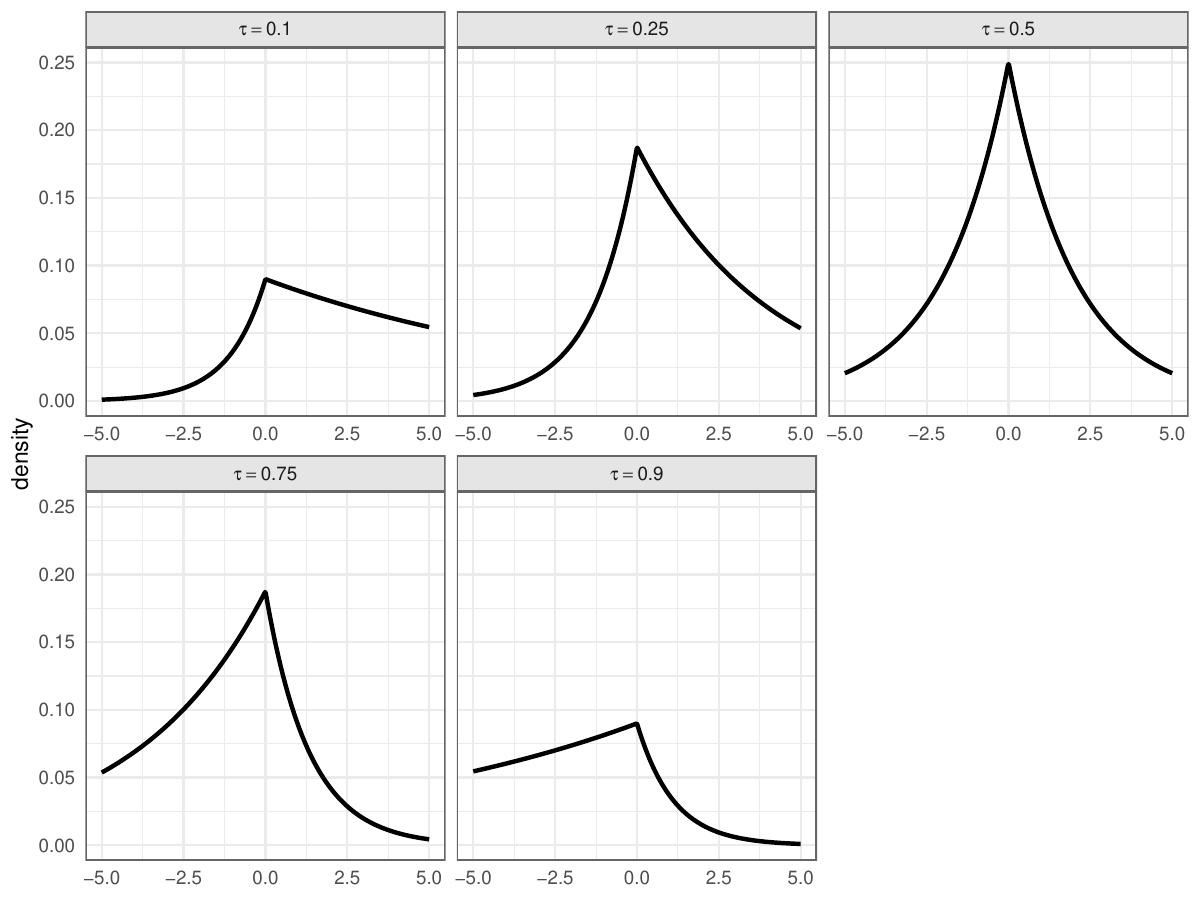}
    \caption{\textbf{Asymmetric Laplace density (ALD) with $\theta = 0$ and $\lambda = 1$ for various quantile levels $\tau$.} The ALD is skewed to the left when $\tau > 0.5$, and skewed to the right when $\tau < 0.5$; when $\tau = 0.5$, the ALD is the same as the Laplace double exponential distribution.}
    \label{fig:ald-density}
\end{figure}

\begin{algorithm}
\caption{Iterative Maximum-Likelihood Estimation of $(\theta, \lambda, \tau)$ for the asymmetric Laplace distribution.}
\label{alg:iterative}
\begin{algorithmic}[1]

\STATE \textbf{Input:} Estimated ratios $\hat r_1, ..., \hat r_p$ with their standard errors $se(\hat r_1),...,se(\hat r_p)$, tolerance $\epsilon > 0$.
\STATE \textbf{Initialize:} $\tau^{(0)}=0.5$, $w_i = se(\hat r_i)^{-1}$, set $k \gets 0$.

\REPEAT

    \STATE Update $\theta$ given updated $\tau^{(k)}$:
    \[
        \theta^{(k)} \gets \underset{\theta}{\arg\min}\sum_{i=1}^p w_i\rho_{\tau^{(k)}}(\hat r_i - \theta)
    \]

    \STATE Update $\lambda$ given $\tau^{(k)}$ and $\theta^{(k)}$:
    \[
        \lambda^{(k)} \gets \frac{p}{\sum_{i=1}^p w_i\rho_{\tau^{(k)}}(\hat r_i - \theta^{(k)})}.
    \]
    
    \STATE Update $\tau$ given current $\lambda^{(k)}$ and $\theta^{(k)}$:
    \[
        \tau^{(k+1)} \gets 0.5 - \frac{a^{(k)}}{2(2p + \sqrt{(a^{(k)})^2 + 4p^2})},
    \]
     where $a^{(k)} = \lambda^{(k)}\sum_{i=1}^p w_i(\hat r_i - \theta^{(k)})$.

    \STATE Update log-likelihood:
    \[
        \ell(\theta^{(k+1)},\tau^{(k+1)},\lambda^{(k+1)}) = \sum_{i=1}^p \log w_i + \sum_{i=1}^p \log \left( \lambda^{(k+1)}\tau^{(k+1)}(1-\tau^{(k+1)}) \right) - \lambda^{(k+1)}\sum_{i=1}^p w_i\rho_{\tau^{(k+1)}}(\hat r_i - \theta^{(k+1)}).
    \]

    \STATE $k \gets k + 1$

\UNTIL{$|\ell(\theta^{(k+1)},\tau^{(k+1)},\lambda^{(k+1)}) - \ell(\theta^{(k)},\tau^{(k)},\lambda^{(k)})| < \epsilon$}

\STATE \textbf{Output:} $\theta^{(k+1)}, \lambda^{(k+1)}, \tau^{(k+1)}$

\end{algorithmic}
\end{algorithm}

\section{Simulation study}
\subsection{Simulation study design}
\subsubsection{Simulations with strong pleiotropic effects}
We compared the performance of our proposed method with existing \ac{MR} estimators through simulations following~\citet{Xue2021} and~\citet{Burgess2020} in which simulated pleiotropic effects are relatively strong. We considered the three following scenarios:
\begin{enumerate}
    \item All genetic variants are valid instruments (no pleiotropy).
    \item A fraction $q \in \{20\%,40\%,60\%\}$ of genetic variants have positive uncorrelated pleiotropic effects $\alpha_i$ (Eq.~\eqref{eq:modelc}), generated from a $\text{Uniform}(0.2,0.3)$ distribution for $i=1,\dots,m$ (uncorrelated pleiotropy).
    \item A fraction $q \in \{20\%,40\%,60\%\}$ of genetic variants have both positive uncorrelated pleiotropic effects $\alpha_i$ (Eq.~\eqref{eq:modelc}), generated from $\text{Uniform}(0.2,0.3)$, and correlated pleiotropic effects $\phi_i$ (Eq.~\eqref{eq:modela}), generated from $\text{Uniform}(-0.1,0.1)$ for $i=1,\dots,m$ (correlated pleiotropy).
\end{enumerate}

In the third simulation scenario, the \ac{InSIDE} assumption which states that pleiotropic effects are not correlated with the SNP-exposure direct effects is violated due to the presence of correlated pleiotropy~\citep{Burgess2017}. We simulated individual-level data for two independent samples of size $n=50,000$, using the first sample to obtain marginal associations between each variant and the outcome $Y$, and the second samples to obtain marginal associations with the exposure $X$. Individual observations for the risk factor $X$, outcome $Y$ and unmeasured confounder $U$ were generated from the following model:
\begin{subequations}\label{eq:model}
\begin{align}
U &= \sum_{i=1}^p \phi_i \cdot G_i + \epsilon_U, \label{eq:modela} \\
X &= \sum_{i=1}^p \gamma_i \cdot G_i + \beta_{XU}\cdot U + \epsilon_X, \label{eq:modelb} \\
Y &= \theta_0 \cdot X + \sum_{i=1}^p \alpha_i \cdot G_i + \beta_{YU}\cdot U + \epsilon_Y . \label{eq:modelc}
\end{align}
\end{subequations}
For each genetic variant $i=1,\dots,p$, with $p=30$ or $100$, we first simulated its \ac{MAF} $f_i$ from a $\text{Uniform}(0.1,0.3)$ distribution. We then generated a vector of SNP genotypes $G_i \in \mathbb{R}^n$, where each entry was simulated from a Binomial$(2,f_i)$. Invalid variants were randomly selected at each replication, with the number of invalid variants given by $m = q \cdot p$. The SNP–exposure direct effects $\gamma_i$ (Eq.~\eqref{eq:modelb}) were simulated from either $\text{Uniform}(0.1,0.2)$ or $\text{Uniform}(-0.2,-0.1)$ with equal probability. We set $\beta_{XU} = \beta_{YU} = 1$ and assumed standard normal errors $\epsilon_U$, $\epsilon_X$, and $\epsilon_Y$ (Eqs.~\eqref{eq:modela}--\eqref{eq:modelc}). The causal effect $\theta_0$ (Eq.~\eqref{eq:modelc}) was set to $0$ (null model), $-0.1$, or $0.1$. A total of $10{,}000$ replications were performed for the null models to evaluate type I error, and $2{,}000$ replications were performed otherwise.

\subsubsection{Simulations with weak invalid IVs}
In the previous simulation design, pleiotropic effects are relatively large compared with the direct SNP–exposure effects, representing an idealized setting for methods that rely on the identification of invalid instruments. We therefore considered additional simulations as in~\citep{Xue2021} in which pleiotropic effects are weaker, making the identification of invalid instruments substantially more challenging. Specifically, we simulated $p = 50$ genetic variants and a sample size of $n = 50{,}000$. The first 30 variants were selected as invalid instruments, corresponding to a proportion $q = 60\%$. The pleiotropic effects of the invalid \acp{IV} were generated in two ways. Uncorrelated pleiotropic effects $\alpha_i$ were drawn from a normal distribution,  
\[
\alpha_i \sim N\left(0, \frac{h_y^2}{p}\right), \quad h_y^2 \in \{0.1, 0.2, 0.4, 0.6\},
\]  
while correlated pleiotropic effects $\phi_i$ were drawn from  
\[
\phi_i \sim N\Big(0, \frac{h_u^2}{p}\Big), \quad h_u^2 \in \{0, 0.1\},
\] for $i=1,...,30.$ Thus, setting $h_u^2 = 0$ ensured that all correlated pleiotropic effects were null. The SNP-exposure direct effects $\gamma_i$, for $i = 1, \dots, p$, were generated as  
\[
\gamma_i \sim N\Big(0, \frac{h_x^2}{p}\Big), \quad h_x^2 = 0.5.
\]  
Instead of simulating individual-level data as in the previous simulation design, SNP-exposure and SNP-outcome summary statistics were generated directly from normal distributions  
\[
\hat\beta_{Xi} \sim N(\gamma_i + \phi_i, \hat\sigma_{Xi}^2), \quad 
\hat\beta_{Yi} \sim N\big(\theta \cdot (\gamma_i + \phi_i) + \alpha_i + \phi_i, \hat\sigma_{Yi}^2\big),
\]  
with $\hat\sigma_{Xi}^2 = \hat\sigma_{Yi}^2 = 1/n$. We considered multiple values for the true causal effect $\theta \in \{-0.2, -0.1, -0.05, 0, 0.05, 0.1, 0.2\}$.

\subsubsection{MR methods}
In all simulations, we compared the performance of our proposed MR-Quantile with the following methods: MR-IVW~\citep{Burgess2013}, MR-Egger~\citep{Burgess2017}, MR-PRESSO~\citep{Verbanck2018}, MR-RAPS~\citep{Zhao2020}, MR-Weighted-Median~\citep{Bowden2016}, MR-Weighted-Mode~\citep{Hartwig2017}, MR-Lasso~\citep{Slob2020}, cML-MA-BIC and cML-MA-BIC-DP~\citep{Xue2021}, MR-Mix~\citep{Qi2019} and MR-ContMix~\citep{Burgess2020}. We used the default implementation provided by the \texttt{TwoSampleMR} package (version 0.6.29) in R for MR-IVW, MR-Egger, MR-Weighted-Median, MR-Weighted-Mode, and MR-RAPS. For MR-LASSO and MR-ContMix, we used the default implementation provided by the \texttt{MendelianRandomization} package (version 0.10.0) in R. The reference R implementations provided by the authors were used for the other estimators. As an optimal unbiased estimator, we also included the MR-IVW method on the subset of valid IVs, denoted as MR-IVW (Oracle). Due to its very high running time relative to the other methods, we did not include MR-CAUSE~\citep{Morrison2020} in this simulation study. Table \ref{tab:mr_methods} provides an overview of the methods and the assumptions required for each method to yield consistent estimates of the causal effect.

\begin{table}[htbp]
\centering
\caption{Overview of different MR methods compared in simulations, including which assumptions each method relies on.}
\label{tab:mr_methods}
\resizebox{\textwidth}{!}{%
\begin{tabular}{lll}
\toprule
\textbf{Method} & \textbf{Model / Approach} & \textbf{Assumptions} \\
\midrule
MR-IVW 
  & Weighted linear regression with inverse variance weights
  & No invalid IV \\
MR-Egger 
  & Weighted linear regression with unconstrained intercept to detect/correct directional pleiotropy 
  & InSIDE \\
MR-RAPS 
  & Robust adjusted profile score accommodating weak instruments and systematic pleiotropy 
  & InSIDE \\
MR-PRESSO 
  & Detects and removes outlier IVs via a global pleiotropy test and outlier correction 
  & InSIDE, Majority valid \\
MR-Weighted-Median 
  & Weighted median of ratio estimates
  & Majority valid \\
MR-Weighted-Mode 
  & Weighted mode of the smoothed empirical density function of ratio estimates
  & Plurality valid \\
MR-Lasso 
  & Weighted linear regression with lasso penalty 
  & Plurality valid \\
cML-MA-BIC 
  & Constrained maximum likelihood with model averaging and BIC selection over invalid IV sets 
  & Plurality valid \\
cML-MA-BIC-DP 
  & Extension of cML-MA-BIC with data perturbation for more consistent selection of invalid IVs
  & Plurality valid \\
MR-Mix 
  & Normal-mixture model identifying valid and invalid IVs through distinct mixture components
  & Plurality valid \\
MR-ContMix 
  & Contamination mixture model where invalid IVs are modeled as diffuse contamination
  & Plurality valid \\
MR-Quantile 
  & Weighted quantile regression based on the asymmetric Laplace distribution
  & Plurality valid \\
\bottomrule
\end{tabular}%
}
\begin{tablenotes}
\scriptsize
\item \textit{Note:} Balanced pleiotropy: the method assumes zero-mean pleiotropic effects. InSIDE: the method requires the Instrument Strength Independent of Direct Effect assumption. Majority valid: the method is consistent when $\geq 50\%$ of IV weights are valid. Plurality valid: the method is consistent when the largest group of IVs sharing the same ratio estimate corresponds to valid IVs.
\end{tablenotes}
\end{table}

\subsection{Results}
\subsubsection{Simulations with strong pleiotropic effects}
Empirical type I error rates for scenario 1 (no pleiotropy) and scenario 2 (uncorrelated pleiotropy) are presented in Supplementary Figure \ref{fig:typeIerror_InSIDE}. When all instruments were valid, all methods controlled the type I error rate adequately except for MR-ContMix. As the fraction of uncorrelated pleiotropic variants increased, the MR-PRESSO method rejected the null well above the nominal level of 0.05, especially when the number of instruments was low. The MR-Weighted-Mode and MR-Mix methods were generally the most conservative, but the former had inflated type I error when 60\% of instruments were invalid and $p=30$, presumably because it is challenging to estimate the mode of a distribution with a small number of valid IVs. In summary, the cML-MA-BIC, cML-MA-BIC-DP, MR-Egger, MR-IVW, MR-Lasso, MR-Weighted-Median, MR-Mix and MR-Quantile methods controlled the type I error rate adequately for all simulation settings when the \ac{InSIDE} assumption held.

Supplementary Figure \ref{fig:typeIerror_noInSIDE} presents the empirical type I error rates for the third simulation scenario with correlated pleiotropy, where the \ac{InSIDE} assumption was violated. As expected, MR-IVW, MR-Egger, MR-PRESSO and MR-RAPS all had inflated type I error rates. Moreover, when the fraction of invalid IVs was equal to 60\%, the rejection rate of MR-Weighted-Median was approximately twice the nominal level, as opposed to MR-Quantile which does not rely on the majority valid assumption. Once again, the MR-Weighted-Mode and MR-Mix methods were the most conservative methods when the fraction of invalid IVs was below 60\%, and MR-Weighted-Mode had the highest type I error inflation when 60\% of instruments were invalid. In summary, the cML-MA-BIC, cML-MA-BIC-DP, MR-Lasso, MR-Mix and MR-Quantile methods controlled the type I error rate adequately for all simulation settings when the \ac{InSIDE} assumption was violated.

Empirical distributions of the MR estimates of the causal effect for scenario 2 (right panels) and scenario 3 (left panels) with 60\% invalid IVs are presented in Figure \ref{fig:boxplots_main} when the true causal effect was either equal to $\theta_0=0.1$ (top panels) or $\theta_0=-0.1$ (bottom panels). In all simulation settings, the variance of MR-Egger, MR-IVW, MR-PRESSO and MR-RAPS was very large compared to other methods. Even when the \ac{InSIDE} assumption held, MR-Weighted-Mode and MR-Weighted-Median estimates were slightly biased towards 0. For MR-Weighted-Median, this is explained by the fact that in finite samples, the weighted median is influenced by the distribution of pleiotropic effects. Indeed, when the distribution of invalid IVs is not balanced around the true causal effect, then the weighted median estimator is shifted away from the weighted median of the valid IVs only. For MR-Weighted-Mode, the observed bias towards 0 is explained by the fact that in small samples, invalid genetic instruments identifying causal effect parameters that are close to the true causal effect are likely to contaminate the estimate of the mode. When the \ac{InSIDE} assumption was violated, MR-Egger, MR-IVW, MR-RAPS and MR-Weighted-Mode were substantially biased and had high variance. Moreover, MR-Weighted-Median and MR-Quantile were slightly negatively biased, due to the residuals being correlated with the SNP-exposure direct effects, although the \ac{RMSE} for MR-Quantile was always lower than for MR-Egger, MR-IVW, MR-Lasso, MR-Weighted-Median, MR-Weighted-Mode and MR-PRESSO. As expected, because the simulated pleiotropic effects were relatively large, methods that rely on consistent selection and removal of invalid IVs, such as cML-MA-BIC, cML-MA-BIC-DP and MR-LASSO, or methods that rely on mixture models, such as MR-ContMix and MR-Mix, had the highest power to detect a causal effect across all simulation settings. 

\begin{figure}[tbp]
    \centering
    \includegraphics[width=\textwidth]{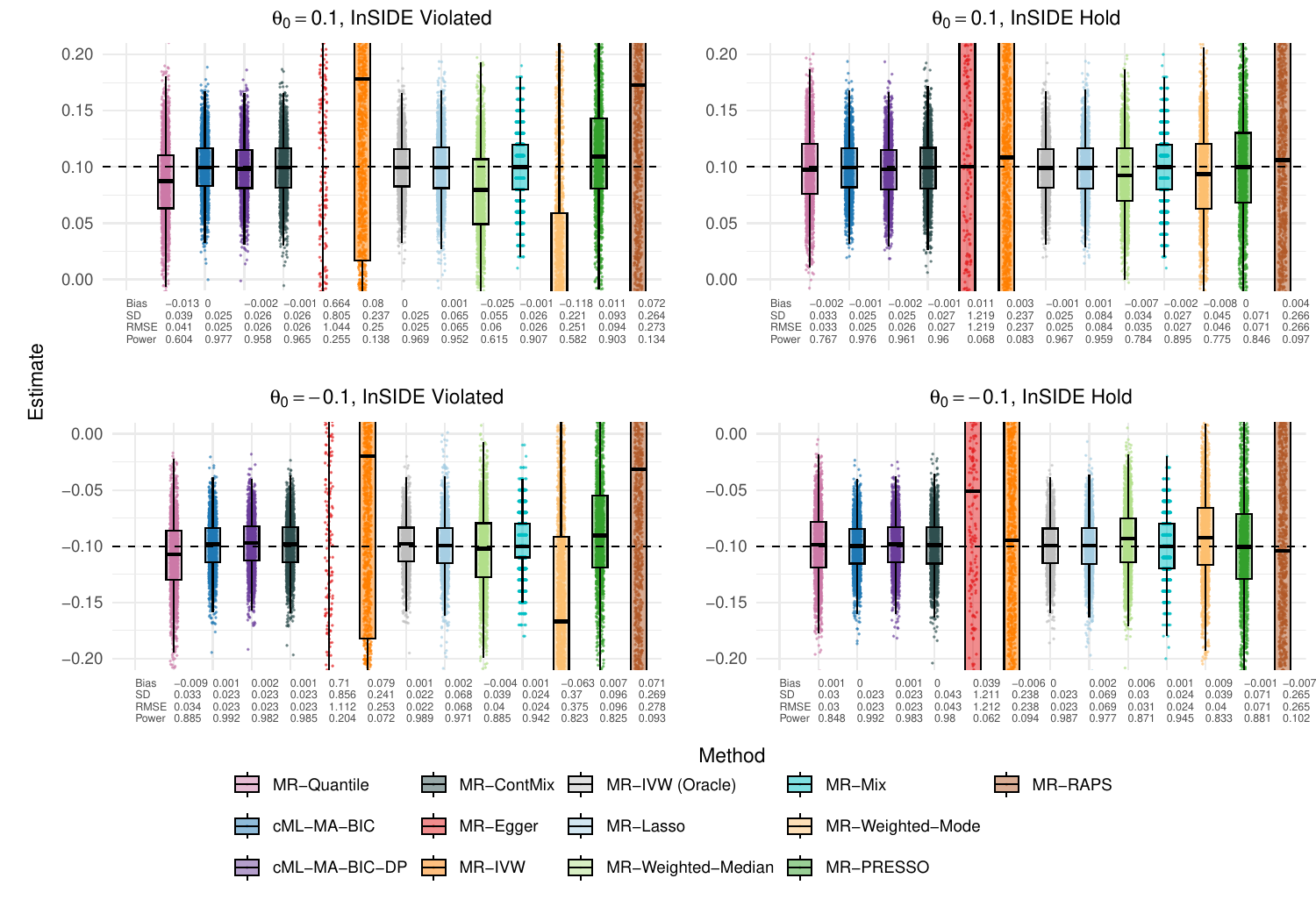}
    \caption{\textbf{Simulations with strong pleiotropic effects.} Empirical distributions of the estimates of the true causal effect $\theta_0 = 0.1$ and $\theta_0 = -0.1$ with sample size $n=50,000$, number of SNPs $p=30$, and fraction of invalid IVs $q=60\%$. Bias = $\E[\hat\theta - \theta_0]$; Standard deviation (SD) = $\sqrt{\E[(\hat\theta - \E[\hat\theta])^2]}$; Root mean square error (RMSE) = $\sqrt{\E[(\hat\theta - \theta_0)^2]}$.}
    \label{fig:boxplots_main}
\end{figure}

Supplementary Table \ref{tab:runtime} reports the computation time (in seconds) for the MR methods over 100 replications under the third simulation scenario, with $\theta_0 = 0.1$, $q = 60\%$ invalid IVs, and $p = 30$ and $p = 100$ SNPs. All simulations were performed on a MacBook Air equipped with an Apple M1 chip and 8 GB of memory. The number of bootstrap replications for MR-Weighted-Median and MR-Quantile was equal to 1000. The mean runtime of MR-Quantile was equal to 0.26 and 0.30 seconds for $p = 30$ and $p = 100$, respectively. Overall, most methods remained computationally efficient as $p$ increased from 30 to 100, except for MR-ContMix, MR-Mix, cML-MA-BIC-DP and MR-PRESSO, which exhibited substantially larger increases in computation time. Notably, MR-ContMix, cML-MA-BIC-DP and MR-PRESSO displayed high variability in runtime at $p=100$ (\ac{SD} of 13.97 seconds, 8.76 seconds and 13.28 seconds respectively), suggesting unstable convergence in some replications. 

\subsubsection{Simulations with weak invalid IVs}
Empirical type I error rates for the simulation design with weak pleiotropic effects are presented in Figure \ref{fig:typeIerror_weak}. As expected, across all simulations, the methods that rely on consistent selection of invalid IVs, such as cML-MA-BIC and MR-Lasso, exhibited inflated type I error rates because it is more challenging to distinguish weak pleiotropic IVs from valid IVs. This is reflected by the decreasing type I error of these methods when the variance of the pleiotropic effects increases, reflecting improved discrimination between valid and invalid instruments. In contrast, cML-MA-BIC-DP, which uses data perturbation to consistently select invalid IVs, was the most conservative method and controlled the type I error adequately. Across all simulations, MR-Quantile was the only other method with low type I error, ranging from 0.059 to 0.072. When the variance $h^2_u$ of the correlated pleiotropic effects was equal to 0, MR-Egger, MR-IVW and MR-RAPS also controlled the type I error satisfactorily, as uncorrelated pleiotropic effects were balanced and normally distributed. On the other hand, in the presence of correlated pleiotropic effects ($h^2_u=0.1$), the type I error rates for all methods, except for cML-MA-BIC-DP and MR-Quantile, were highly inflated, especially in the setting with $h^2_y = h^2_u = 0.1$.

\begin{figure}[tbp]
    \centering
    \includegraphics[width=\textwidth]{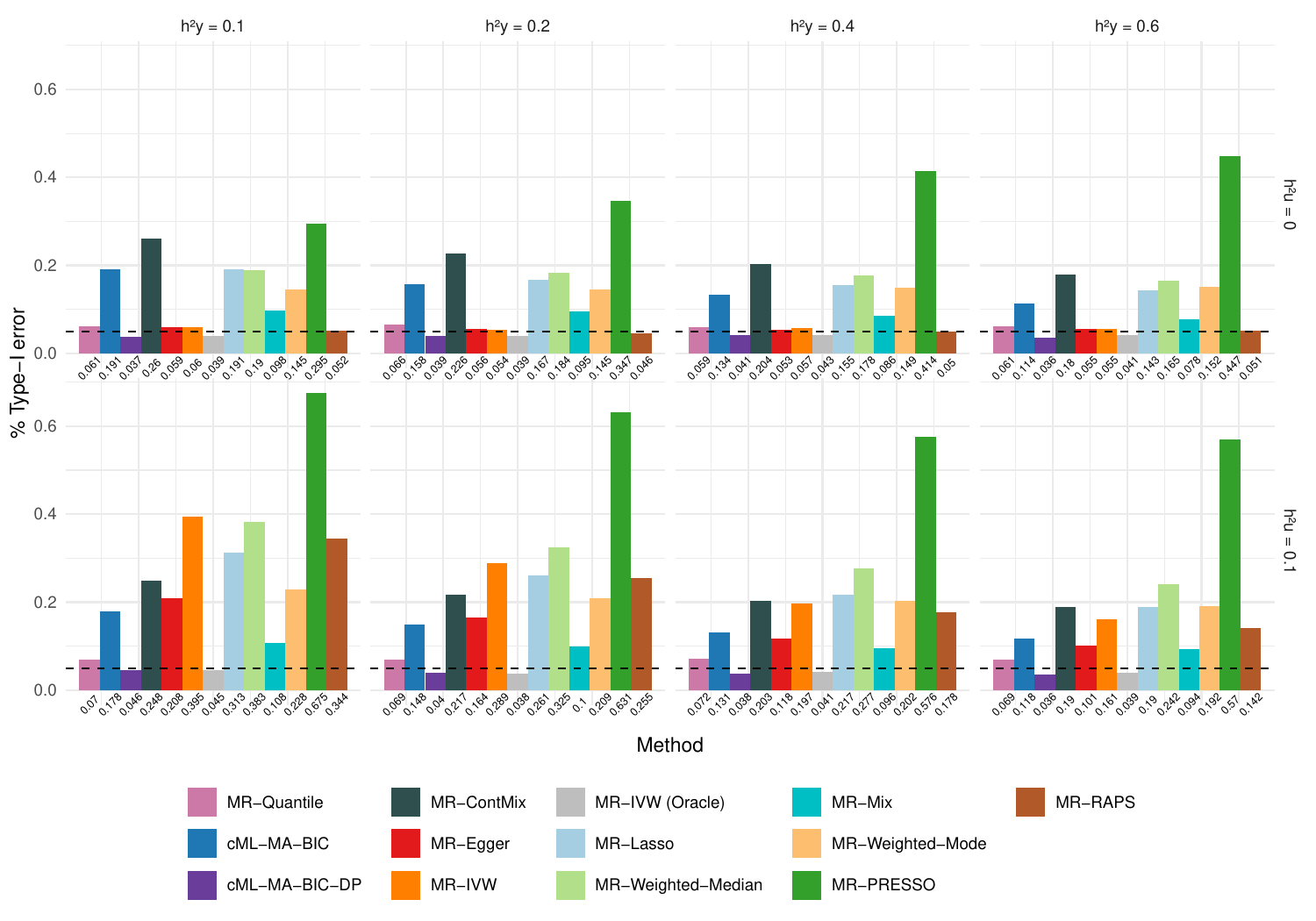}
    \caption{\textbf{Simulations with weak invalid IVs.} Empirical type I error rates at the nominal level of 0.05 with sample size $n=50,000$ and with $p=50$ SNPs. The fraction of invalid IVs is equal to $60\%$, and InSIDE assumption holds when $h^2_u = 0$ (top row), and is violated when $h^2_u = 0.1$ (bottom row). The strength of the uncorrelated and correlated pleiotropic effects increases respectively with $h^2_y$ (left to right) and $h^2_u$ (top to bottom).}
    \label{fig:typeIerror_weak}
\end{figure}

Empirical distributions of the MR methods estimates of the causal effect are presented in Figure \ref{fig:boxplots_weak} for various values of the true causal effect $\theta_0$ when $h^2_y=h^2_u=0.1$. Methods that rely on the \ac{InSIDE} and/or majority valid assumptions all exhibited substantial bias and variance. Moreover, some methods that rely on the plurality valid assumptions, such as MR-Lasso and MR-Weighted-Mode, were also biased in the presence of invalid IVs with weak pleiotropic effects. Of note, when the true causal effect $\theta_0$ was negative, cML-MA-BIC-DP had small to moderate bias towards 0, and lower power to detect a causal effect compared to other unbiased methods such as cML-MA-BIC, MR-ContMix, MR-Mix and MR-Quantile. Overall, MR-Quantile performed comparably well across all settings in term of bias, variance and power to detect a significant causal effect, as well as maintaining type I error rates close to the nominal level of 0.05 across all simulations. 

\begin{figure}[tbp]
    \centering
    \includegraphics[width=\textwidth]{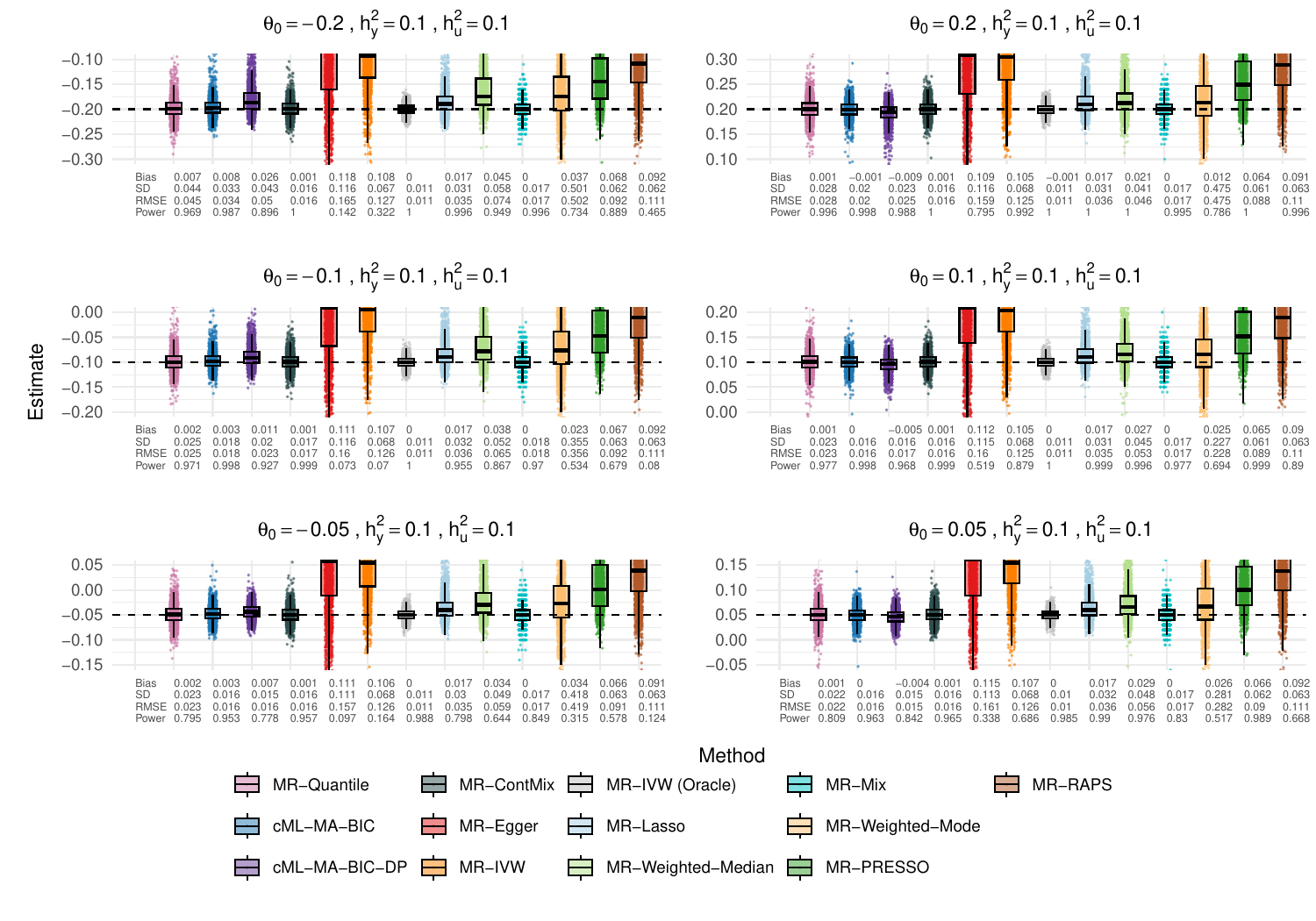}
    \caption{\textbf{Simulations with weak invalid IVs.} Empirical distributions of the estimates of the true causal effect $\theta_0$ with sample size $n=50,000$, number of SNPs $p=50$, and fraction of invalid IVs $q=60\%$. Bias = $\E[\hat\theta - \theta_0]$; Standard deviation (SD) = $\sqrt{\E[(\hat\theta - \E[\hat\theta])^2]}$; Root mean square error (RMSE) = $\sqrt{\E[(\hat\theta - \theta_0)^2]}$.}
    \label{fig:boxplots_weak}
\end{figure}

\section{Causal effect of resting heart rate on atrial fibrillation}

Atrial fibrillation (AF) is a common arrhythmia characterized by the disorganized contraction of the atria. Between 2010 and 2019, the global prevalence of \ac{AF} nearly doubled, increasing from 33.5 million to 59 million individuals~\citep{Linz2024}. When \ac{AF} is left untreated, it can lead to serious complications including cardioembolic stroke. Despite improved monitoring opportunities offered by the advent of consumer wearables and improvements in pharmacotherapy, \ac{AF} remains associated with a significant burden of morbidity and mortality~\citep{Joglar2024}. \ac{AF} pharmacotherapy has focused on two strategies: rate control, where the ventricular rate is controlled without attempting to restore sinus rhythm, or rhythm control, aimed at recovering and maintaining sinus rhythm. While there is now clinical evidence favoring rhythm control, investigating the causal effect of heart rate on \ac{AF} is an interesting question amenable to MR investigations~\citep{Prystowsky2022}. Moreover, observational studies have suggested a U-shaped relationship between \ac{RHR} and \ac{AF}, whereas previous MR analyses have reported an inverse causal effect~\citep{Klevjer2023,vandeVegte2023}. This paradoxical effect is discordant with the use of heart rate control for the treatment of \ac{AF}.

\begin{figure}[tbp]
\includegraphics[width=\textwidth]{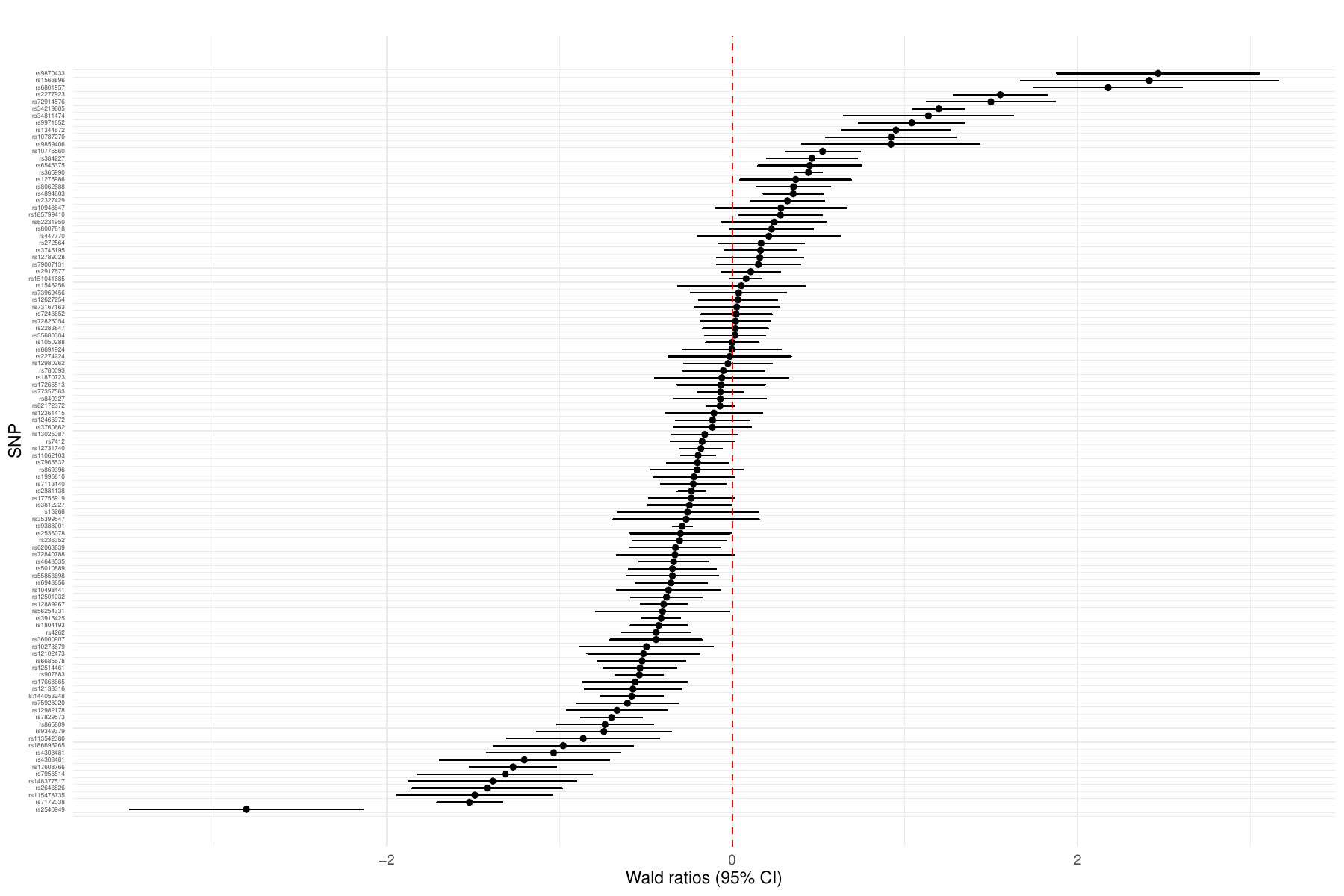}
\caption{\textbf{MR ratio estimates with 95\% CI for the causal effect between \ac{RHR} and AF for 104 independent SNPs associated with \ac{RHR}.}}\label{fig:case_study}
\end{figure}

In this case study, we sought to revisit previous MR findings using the most recent summary statistics for \ac{RHR} and \ac{AF}, and to evaluate them across a broad range of robust MR methods. Genetic associations with \ac{RHR}, after rank-based inverse-normal-transformation, were obtained from 425,748 European ancestry individuals from the VA Million Veteran Program, including approximately $\text{44 million}$ imputed SNPs~\citep{Verma2024}. A total of 425 independent SNPs were significantly associated with \ac{RHR} after correcting for multiple testing using a genome-wide significance threshold of $p < 5 \times 10^{-8}$. Genetic associations with \ac{AF}, expressed as log-\acp{OR}, were obtained from a European meta-analysis which included 228,926 \ac{AF} cases from eight studies and a total of $\text{36 M}$ imputed SNPs~\citep{Yuan2025}. After harmonizing variants between the two GWAS, only 104 of the 425 genome-wide significant SNPs were retained as instruments, as the remaining SNPs could not be found in the GWAS for \ac{AF} which used an older imputation panel with limited coverage. The ratio estimates derived from the 104 IVs were highly heterogeneous, suggesting a high proportion of pleiotropic variants (Figure~\ref{fig:case_study}). Given the low prevalence of \ac{AF} in the European population (1.63\%)~\citep{Linz2024}, the causal log-RR between \ac{RHR} and \ac{AF} can be approximated using ratio estimates based on the reported log-\acp{OR} from \citet{Yuan2025}.

We fitted MR-Quantile using ratio estimates obtained from the 104 independent IVs, and plotted the fitted $ALD(0, \hat\tau, 1)$ density of the standardized residuals $\hat e_i = w_i\hat\lambda(\hat r_i - \hat\theta)$, for $i=1,...,104$ (Figure \ref{fig:fitted_density}). The \acp{MLE} for the model parameters were equal to $\hat\theta=-0.23$, $\hat\tau = 0.42$, and $\hat\lambda=0.61$, corresponding to a RR (95\% CI) of 0.796 (0.728-0.870) (Figure~\ref{fig:forest}). In summary, one standard deviation increase in inverse normal transformed \ac{RHR} was associated with a 10-20\% lower risk of \ac{AF} across all methods summarized in Table \ref{tab:mr_methods}, suggesting a protective effect of increased \ac{RHR} on \ac{AF}, which is consistent with previous MR studies. A possible explanation for the paradoxical observation that increasing \ac{RHR} could protect against \ac{AF} is that some IVs may have context-specific pleiotropic effects. For example, some of the included IVs could have developmental effects on the heart, leading to changes in \ac{RHR}, but also other anatomical or physiological features of the heart. The observed heterogeneity in the MR ratios (Figure \ref{fig:case_study}) suggests that IVs may cluster into distinct modes, potentially reflecting subgroups with different underlying biological mechanisms or latent pleiotropic pathways. 

As an exploratory analysis aimed at identifying potential pleiotropic pathways biasing MR estimates, we performed a phenome-wide association study (PheWAS, Supplementary Table \ref{tab:opengwas}) using the OpenGWAS database \citep{Elsworth2020} and report in Supplementary Figure \ref{fig:heatmap} the most significant associations between the SNPs used as IVs and traits from different phenotype categories. We observed that SNPs with lower ratio estimates for the causal effect of \ac{RHR} on \ac{AF} tended to be positively associated with traits such as body mass index, impedance, weight, basal metabolic rate, forced expiratory volume, and forced vital capacity. Many of these traits are closely related to cardiovascular phenotypes. For example, higher body mass and basal metabolic rate are generally associated with higher \ac{RHR}, while pulmonary function traits such as improved forced expiratory volume and vital capacity are linked to greater cardiopulmonary efficiency and enhanced autonomic regulation, which together act to lower RHR and reduce arrhythmic risk. Clustering of invalid IVs into distinct groups may potentially violate the plurality valid assumption on which most MR methods rely on to consistently identify invalid IVs and/or identify the true causal effect. Future work exploring these potential clusters may help disentangle valid causal effects from pleiotropic effects and provide a clearer understanding of the pathways linking \ac{RHR} to \ac{AF}.

\begin{figure}
\centering
\includegraphics[scale=0.5]{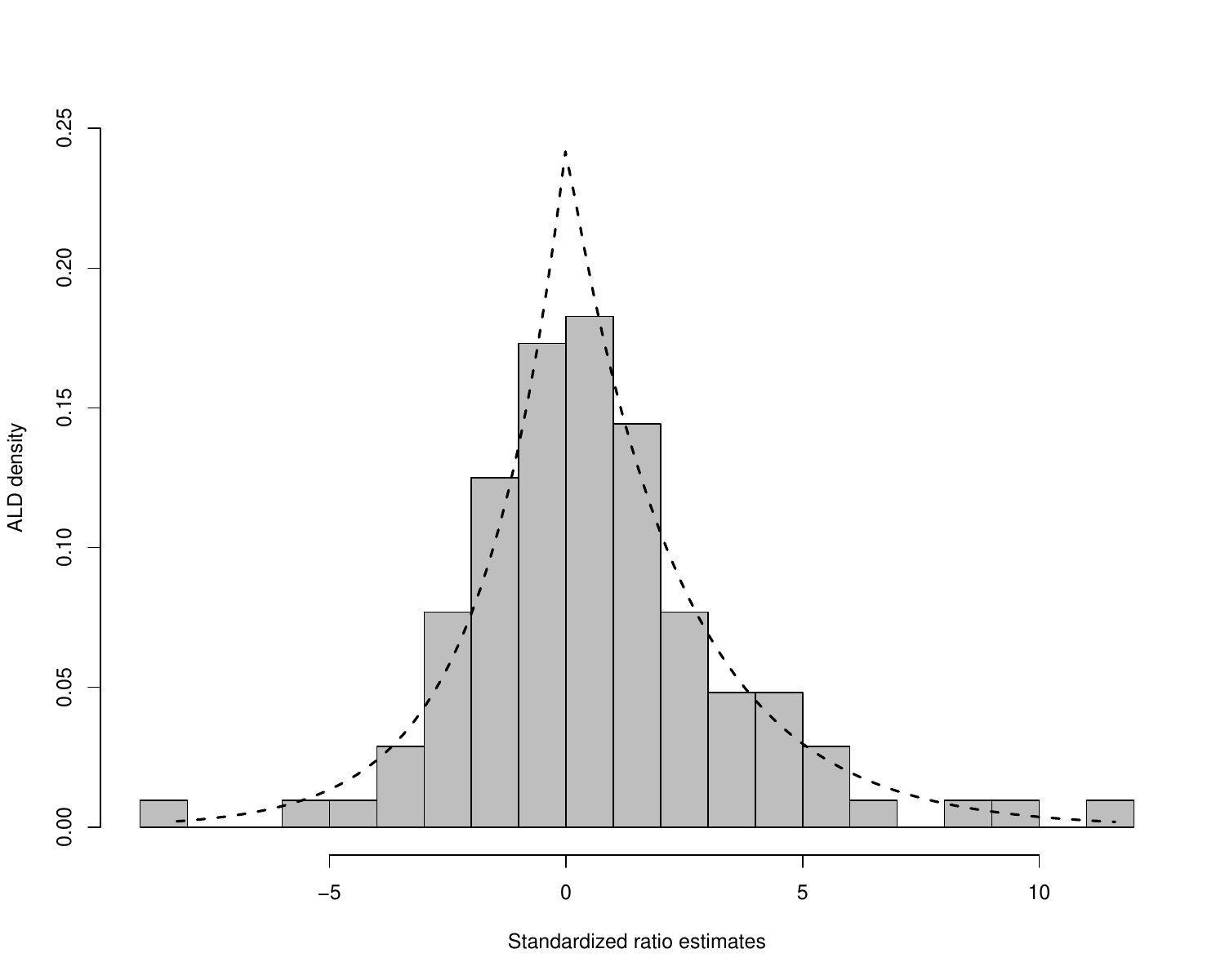}
\caption{\textbf{Fitted $\boldsymbol{ALD(0, \hat\tau, 1)}$ density of the standardized residuals $\boldsymbol{\hat e_i = w_i\hat\lambda(\hat r_i - \hat\theta)}$ for the causal effect between resting heart rate and atrial fibrillation for 104 independent SNPs.} The \acp{MLE} for the model parameters were equal to $\hat\theta=-0.23$, $\hat\tau = 0.42$, and $\hat\lambda=0.61$.}\label{fig:fitted_density}
\end{figure}

\begin{figure}
\centering
\includegraphics[
    scale=1,
    trim=4cm 2.5cm 4cm 3cm,
    clip
]{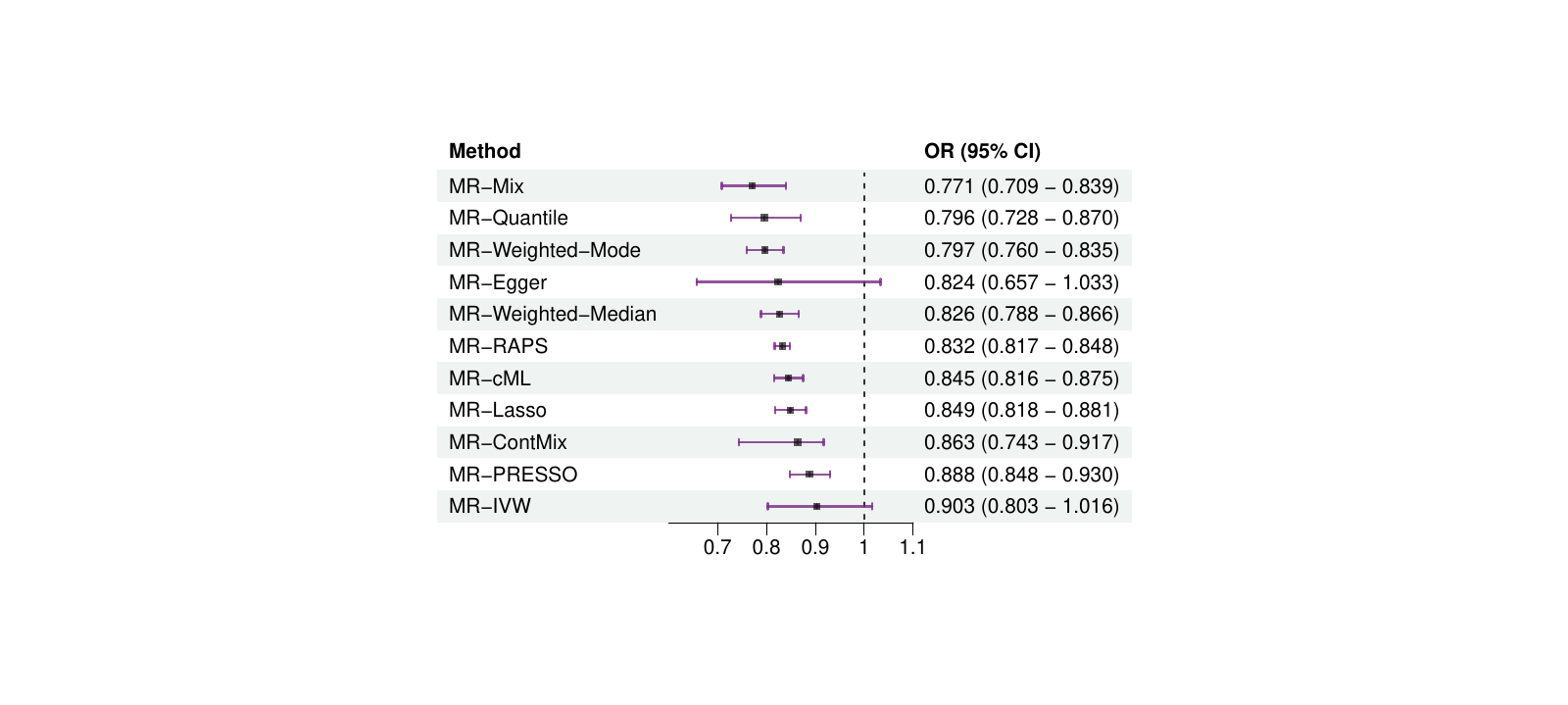}
\caption{\textbf{Estimated causal effect of resting heart rate on atrial fibrillation for compared MR methods, expressed as relative risk with 95\% confidence intervals.}}\label{fig:forest}
\end{figure}

\section{Discussion}

We have proposed a new method based on \ac{WQR} called MR-Quantile to perform consistent estimation and inference of the causal effect between an exposure and outcome of interest in the presence of both uncorrelated and correlated pleiotropic effects. We assumed an \ac{ALD} for the ratio estimates obtained from summary statistics for the SNP-exposure and SNP-outcome relationships for both valid and invalids IVs, and showed that the resulting \ac{MLE} for the location parameter was equivalent to a weighted quantile estimator, with weights inversely proportional to the standard error of the ratio estimates. In contrast to the weighted median estimator~\citep{Bowden2016} which is not consistent when more than 50\% of the instruments weights come from invalid IVs, we have shown that our estimator is consistent for the causal effect if the estimated quantile of the ratio estimates lies within the mode corresponding to the valid instruments, similarly to other mode-based estimators. We leveraged the likelihood-based formulation of the proposed model to perform data-driven estimation of the optimal quantile level $\tau$ that does not rely on any assumption on the true proportion of invalid variants in the observed sample.

We performed extensive simulations with both strong and weak pleiotropic effects, and showed that our proposed method performed comparably well in controlling the type I error rate and \ac{RMSE} in various settings, especially when the proportion of invalid IVs was high. In simulations with strong correlated pleiotropic effects, our proposed method performed competitively and was on par with other methods relying on identification of invalid IVs (cML-MA-BIC, cML-MA-BIC-DP, MR-Lasso) or on mixture models (MR-Mix). In simulations with weak correlated pleiotropic effects, cML-MA-BIC-DP was the only method that consistently controlled the type I error rate at the nominal level across all settings, due to the use of both data perturbation and model averaging, which came at the expense of high computational cost. While all other methods exhibited substantial type I error inflation, MR-Quantile produced type I error rates close to the nominal level. In addition, we showed that MR-Quantile had lower bias and higher power to detect a significant causal effect than cML-MA-BIC-DP across simulations when both correlated and uncorrelated pleiotropic effects were weak. Moreover, MR-Quantile was considerably more computationally efficient than cML-MA-BIC-DP, nearly 120 times faster when the number of SNPs was equal to 100, making it particularly well-suited for MR studies with a large number of IVs. Although in MR studies the number of instruments rarely exceeds 100, scalability becomes critical in settings where MR estimation is done repeatedly. For example, transcriptome-wide association studies (TWAS) are increasingly used to test the association between gene expression and phenotypes using MR models~\citep{Zhu2016}. In TWAS analyses, it is commonplace to include multiple tissues and genes, routinely leading to the estimation of $>$100,000 causal effects via MR. In such contexts, even moderate differences in runtime can translate into substantial cumulative computational costs.

We replicated recent MR studies on the causal effect of \ac{RHR} on \ac{AF} using summary statistics from two large GWAS, and found a protective effect of \ac{RHR} on the risk of \ac{AF}. As reported by previous studies, it is unlikely that increasing \ac{RHR} above the normal adult range (60-100 beats per minute (bpm)) would confer a protective effect against \ac{AF}.~\citet{Siland2022} applied nonlinear MR to estimate the association between \ac{RHR} and the logarithm of the incident \ac{AF} hazard rate by stratifying participants according to instrumental variable–free \ac{RHR} $(<65, 65-75, >75 \text{bpm})$. They found an inverse causal association among individuals with RHR below 65 bpm, with no evidence of association at higher RHR levels. While these findings may appear more biologically plausible than assuming a linear inverse causal effect across the full range of \ac{RHR}, increasing concerns have been raised on the use of nonlinear MR methods and the interpretation from their application~\citep{Wade2023,Burgess2023,Smith2023}. Moreover, both nonlinear and linear MR models assume that genetic effects remain constant across levels of the exposure and covariates. Alternatively, nonparametric IV estimators that allow flexible functional forms between the IV and exposure and between the exposure and outcome should be explored, although they require the use of individual-level data~\citep{Legault2025,Hartford2017,He2023}.

There are some limitations to our proposed method. First, since the density of the \ac{ALD} is non-differentiable at the mode, standard maximum likelihood theory is not directly applicable to derive the asymptotic variance of the \acp{MLE}. If the skewness parameter $\tau$ (quantile level) was assumed to be known, then the \acp{MLE} for the location parameter $\theta$ and inverse scale parameter $\lambda$ can be shown to be asymptotically normal since both estimators are expressed as linear combinations of order statistics~\citep{Kotz2002}.~\citet{Yu2005} derived the asymptotic normality of the \acp{MLE} when all parameters are unknown, albeit only in the homoskedastic setting. To maintain robustness in finite samples, we adopted the parametric bootstrap strategy proposed by~\citet{Bowden2016} for the weighted median estimator. This may partially explain why in simulations with strong pleiotropic effects, MR-Weighted-Median and MR-Quantile have less power to detect a significant causal effect compared to methods that rely on asymptotic normality. 

Second, our proposed MR estimator is currently limited to inference on the effect of a single exposure on the outcome. Deriving a multivariable extension would be useful for incorporating multiple exposures simultaneously in the model. However, this would require modeling SNP-outcome associations as responses instead of ratio estimates, and the \ac{ALD} may be less well suited for modeling the errors from SNP–outcome summary statistics than for capturing the error structure of ratio estimates. Indeed, the distribution of ratio estimates naturally align with the asymmetric and heavy-tailed behavior of the \ac{ALD}, whereas SNP–outcome associations are asymptotically normally distributed. Finally, we have assumed throughout this work that summary statistics for the SNP-exposure and SNP-outcome associations are independent. It would be of interest for future work to extend the proposed model to GWAS studies with overlapping samples, and also to adapt the methodology to studies with multiple correlated SNPs. 


\subsection*{Author Contributions}
Julien St-Pierre (J.S.) designed the study.~J.S. implemented the method, performed the simulations, analyzed the data, and interpreted the results.~J.S. prepared the initial manuscript draft.~Marc-André Legault (M.L.) and Mireille Schnitzer (M.S.) co-directed the project.~Archer Y. Yang contributed to improving the numerical stability and implementation of the algorithm. All authors reviewed and edited the manuscript.

\subsection*{Acknowledgments}
This study was supported in part by a Discovery Grant from the Natural Sciences and Engineering Research Council of Canada (Grant No. RGPIN-2021-03019 to M.S.), by the Canada Research Chairs program (Grant No. CRC-2024-00009 to M.S.), and by funding from IVADO and the Canada First Research Excellence Fund to J.S and M.L.


\subsection*{Conflicts of Interest}
The authors declare no conflicts of interest.

\subsection*{Data Availability Statement}
The data and code can be found on \textbf{GitHub}: 
\href{https://github.com/julstpierre/mr-quantile}{\textcolor{blue}{https://github.com/julstpierre/mr-quantile}}.

\clearpage
\appendix
\renewcommand{\theequation}{A.\arabic{equation}}
\setcounter{equation}{0}

\subsection*{A.1 Consistent estimation of the causal effect}

Under the standard assumptions, for $i=1,...,p$,
\begin{align*}
\textrm{plim }(\hat \beta_{Yi}) &= \theta_0\beta_{Xi} + \alpha_i + \beta_{YU}\phi_i, \\
\textrm{plim }(\hat \beta_{Xi}) &= \beta_{Xi},
\end{align*}
where $\theta_0$ denotes the causal effect, interpreted either as the average causal effect or the causal log relative risk. Hence,
$$\textrm{plim }(\hat r_i) = \frac{\theta_0\beta_{Xi} + \alpha_i + \beta_{YU}\phi_i}{\beta_{Xi}} = \theta_0 + \frac{\alpha_i + \beta_{YU}\phi_i}{\beta_{Xi}}.$$

Let $\hat \theta := \hat Q_{\tau}(\hat{\textbf{r}})$ be the weighted $\tau$-quantile of the \ac{MR} ratios $\hat r_1, ...,\hat r_p$ with strictly positive weights $w_1,...,w_p$. Using a continuity theorem, it follows that
\begin{align*}
\textrm{plim }(\hat \theta) &= \textrm{plim }(\hat Q_{\tau}(\hat{\textbf{r}})) \\
&= \hat Q_{\tau}(\textrm{plim }(\hat{\textbf{r}})) \\
&= \hat Q_{\tau}(\theta_0 + {\bm{\delta}}) \\
&= \theta_0 + \hat Q_{\tau}({\bm{\delta}}),
\end{align*}
where ${\bm \delta}$ is a vector of length $p$ with elements $\delta_i=\frac{\alpha_i + \beta_{YU}\phi_i}{\beta_{Xi}}$ for $i=1,...,p$, and $\delta_i=0$ for valid \acp{IV} only. Thus, $\hat\theta$ is a consistent estimator of $\theta_0$ if $ \hat Q_{\tau}({\bm{\delta}})=0$. Given strictly positive weights $w_i$ for $i=1,...,p$, the weighted $\tau$-quantile of $\bm \delta$ is the smallest $\delta_{(j)}$ such that $$\sum_{i:\delta_i \le \delta_{(j)}}w_i \ge \tau \sum_{i=1}^p w_i.$$ Hence, $\hat Q_{\tau}({\bm{\delta}})=0$ if $\tau$ satisfies
\begin{align*}
\frac{\sum_{i:\delta_i < 0}w_i }{\sum_{i=1}^p w_i}\le \ &\tau \le \frac{\sum_{i:\delta_i \le 0}w_i }{\sum_{i=1}^p w_i} \\
\Longleftrightarrow q_{-} \le \ &\tau \le 1 -q_{+},
\end{align*}
where $q_{-}$ and $q_{+}$ are respectively the weighted proportion of invalid \acp{IV} with negative and positive pleiotropic effects $\delta_i$.     

\subsection*{A.2 MLE for $\tau$}
The partial derivative of $\ell(\theta, \tau, \lambda)$ with respect to $\tau$ is given by
\begin{align}\label{eq:tau1}
\frac{\partial \ell(\theta, \tau, \lambda)}{\partial \tau} = p \frac{1 - 2\tau}{\tau(1-\tau)} - \lambda\sum_{i=1}^p w_i(\hat r_i - \theta)
\end{align}
and setting it equal to zero and multiplying both sides by $\tau(1-\tau)$ leads to
\begin{align}\label{eq:root}
a\tau^2  - \tau (2 p + a) + p = 0,
\end{align}
with $a = \lambda\sum_{i=1}^p w_i(\hat r_i - \theta)$. The roots of the quadratic equation are given by:
\begin{align*}
\hat \tau_{\pm} = \frac{a + 2p \pm \sqrt{a^2 + 4p^2}}{2a}.
\end{align*}

\begin{proposition}
Let $a \ne 0$, $p>0$. The negative root $\hat \tau_{-}$ is the only valid solution with $0 < \hat \tau_{-} < 1$.
\end{proposition}
\begin{proof}
The condition for $\hat \tau_{\pm}$ being a valid solution is given by $$-1 < \frac{2p \pm \sqrt{a^2 + 4p^2}}{a} < 1.$$
If $a>0$, we can rewrite the previous inequality as
\begin{align*}
-a-2p < \pm\sqrt{a^2 + 4p^2} < a-2p.
\end{align*}
Since $\sqrt{a^2 + 4p^2} > a$, the condition $\sqrt{a^2 + 4p^2} < a - 2p$ can never be satisfied, hence $\hat \tau_{+}$ is not a valid solution. The negative root $\hat \tau_{-}$ is a valid solution if and only if
\begin{align*}
a + 2p > \sqrt{a^2 + 4p^2} > 2p - a
\end{align*}
which is always true for $a > 0$.

\noindent If $a<0$, $\hat \tau_{\pm}$ is a valid solution if and only if
\begin{align*}
-a-2p > \pm\sqrt{a^2 + 4p^2} > a-2p,
\end{align*}
but since $\sqrt{a^2 + 4p^2} > |a| \ge -a$, we have that $-a -2p > \sqrt{a^2 + 4p^2}$ is never satisfied, and $\hat\tau_{+}$ is not a valid solution. The negative root $\hat \tau_{-}$ is a valid solution if and only if
\begin{align*}
a + 2p < \sqrt{a^2 + 4p^2} < 2p - a
\end{align*}
which is always true when $a<0.$
\end{proof}

To avoid numerical instability and division-by-zero errors when $a=0$, we multiply the second term of the valid root $\hat \tau_{-}$ by the conjugate $(2p + \sqrt{a^2 + 4p^2})$:
\begin{align*}
\hat \tau &= \frac{1}{2} + \frac{2p - \sqrt{a^2 + 4p^2}}{2a} \\
&= \frac{1}{2} + \frac{(2p - \sqrt{a^2 + 4p^2})(2p + \sqrt{a^2 + 4p^2})}{2a(2p + \sqrt{a^2 + 4p^2})} \\
&= \frac{1}{2} + \frac{4p^2 - (a^2 + 4p^2)}{2a(2p + \sqrt{a^2 + 4p^2})} \\
&= 0.5 - \frac{a}{2(2p + \sqrt{a^2 + 4p^2})}.
\end{align*}
This algebraically equivalent and continuous formula natively evaluates to exactly $0.5$ when $a = 0$, resolving the singularity.

\begin{proposition}
Let $\bar r_w = \sum_{i=1}^p w_i \hat r_i /\sum_{i=1}^p w_i$ be the weighted sample mean of the estimated ratios $\hat r_1,...,\hat r_p$ with $w_1, ...,w_p > 0$. Then, 
\begin{align*}
\begin{cases}
\hat \tau < 0.5 &\text{ if }\bar r_w > \theta, \\
\hat \tau = 0.5 &\text{ if }\bar r_w = \theta, \\
\hat \tau > 0.5 &\text{ if }\bar r_w < \theta.
\end{cases}
\end{align*}
\end{proposition}
\begin{proof}
The \ac{MLE} for $\tau$ is given by the characteristic equation \eqref{eq:root} with $a = \lambda\sum_{i=1}^p w_i(\hat r_i - \theta)$.  First, we can rewrite $a = \lambda\sum_{i=1}^p w_i (\bar r_w - \theta)$ such that $a>0$ if and only if $\bar r_w > \theta$, $a=0$ if and only if $\bar r_w = \theta$, and $a<0$ if and only if $\bar r_w < \theta$. Using the numerically stable formula for $\hat\tau$:
$$ \hat \tau = 0.5 - \frac{a}{2(2p + \sqrt{a^2 + 4p^2})} $$
Since the denominator $2(2p + \sqrt{a^2 + 4p^2})$ is strictly positive, the sign of $(\hat \tau - 0.5)$ is entirely determined by the sign of $-a$. Thus, $\hat \tau = 0.5$ when $a=0$ (i.e., $\bar r_w = \theta$), and $\textrm{sign}(\hat \tau - 0.5) = -\textrm{sign}(a) = -\textrm{sign}(\bar r_w - \theta)$.
\end{proof}

\renewcommand\refname{References}
\bibliographystyle{apalike}
\bibliography{biblio}

\clearpage
\section*{Supplementary Materials}

\subsection*{Supplementary Figures}

\setcounter{figure}{0}
\renewcommand{\thefigure}{S\arabic{figure}}

\begin{figure}[htbp]
    \centering
    \includegraphics[width=\textwidth]{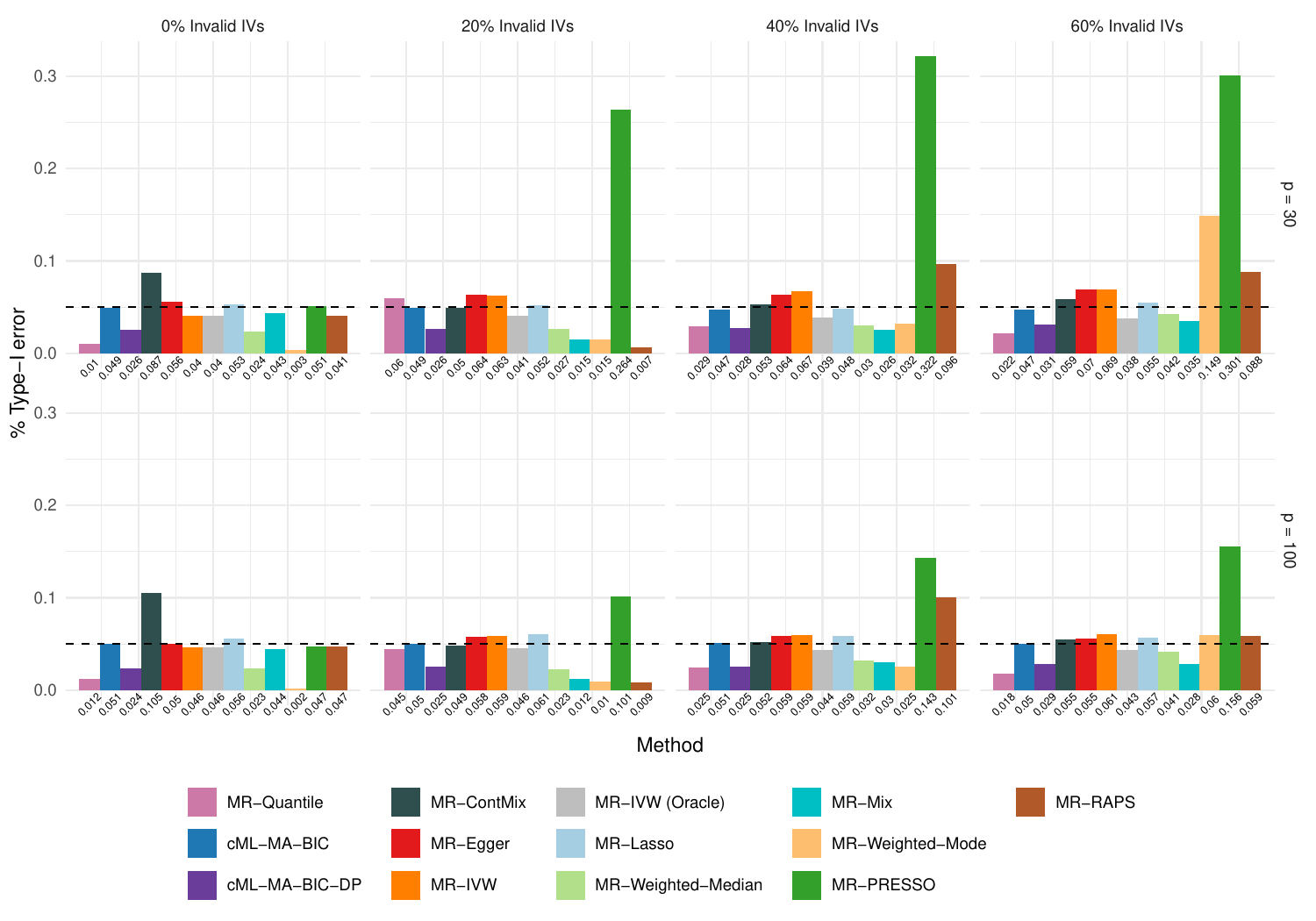}
    \caption{\textbf{Simulations with strong uncorrelated pleiotropic effects (InSIDE assumption holds).} Empirical type I error rates at the nominal level of 0.05 with sample size $n=50,000$ and with $p=30$ or $p=100$ SNPs. The fraction of invalid IVs varies from $0\%$ to $60\%$.}
    \label{fig:typeIerror_InSIDE}
\end{figure}

\begin{figure}[tbp]
    \centering
    \includegraphics[width=\textwidth]{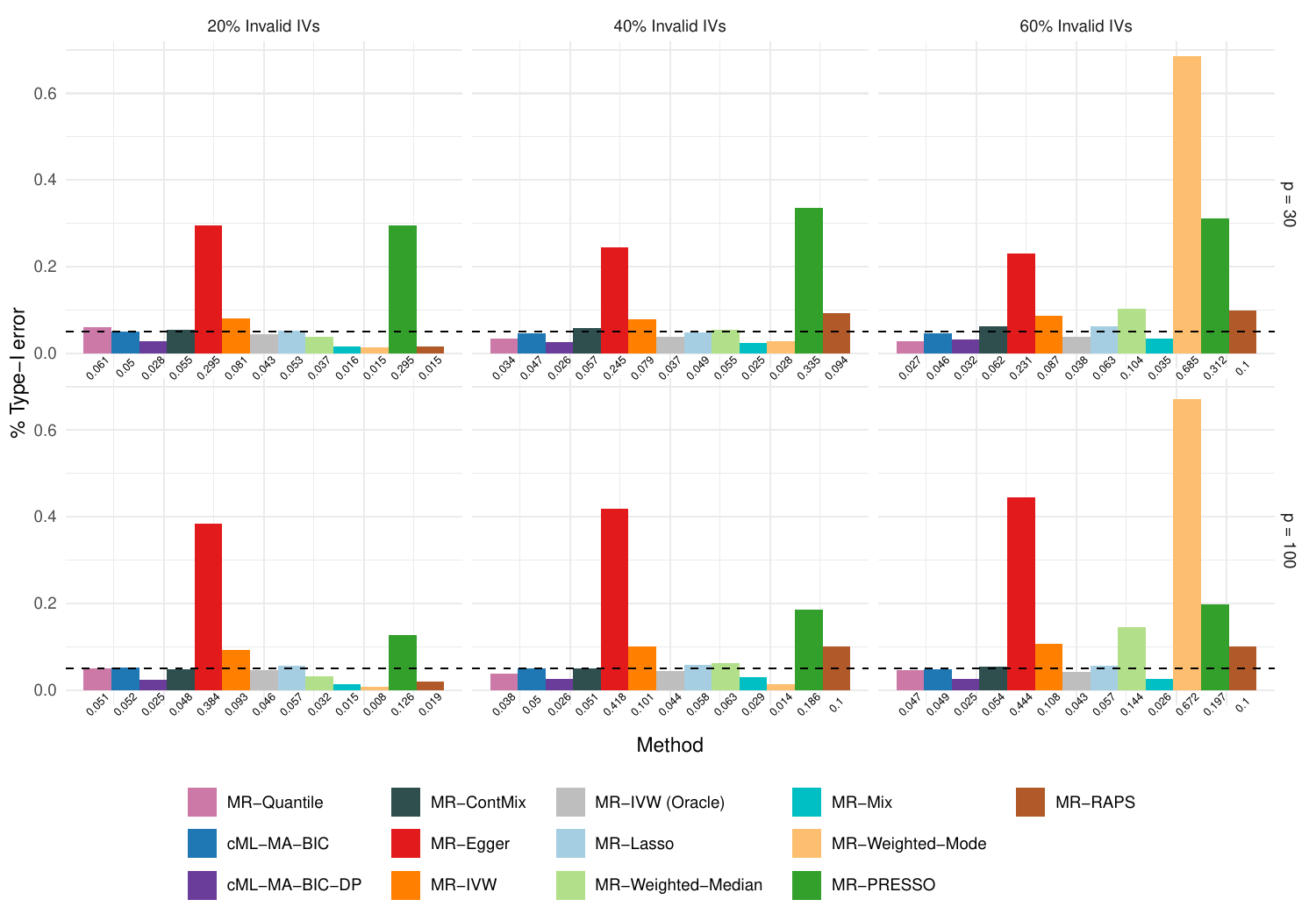}
    \caption{\textbf{Simulations with strong correlated pleiotropic effects (InSIDE assumption violated).} Empirical type I error rates at the nominal level of 0.05 with sample size $n=50,000$ and with $p=30$ or $p=100$ SNPs. The fraction of invalid IVs varies from $20\%$ to $60\%$.}
    \label{fig:typeIerror_noInSIDE}
\end{figure}

\begin{figure}
\centering
\includegraphics[scale=0.7]{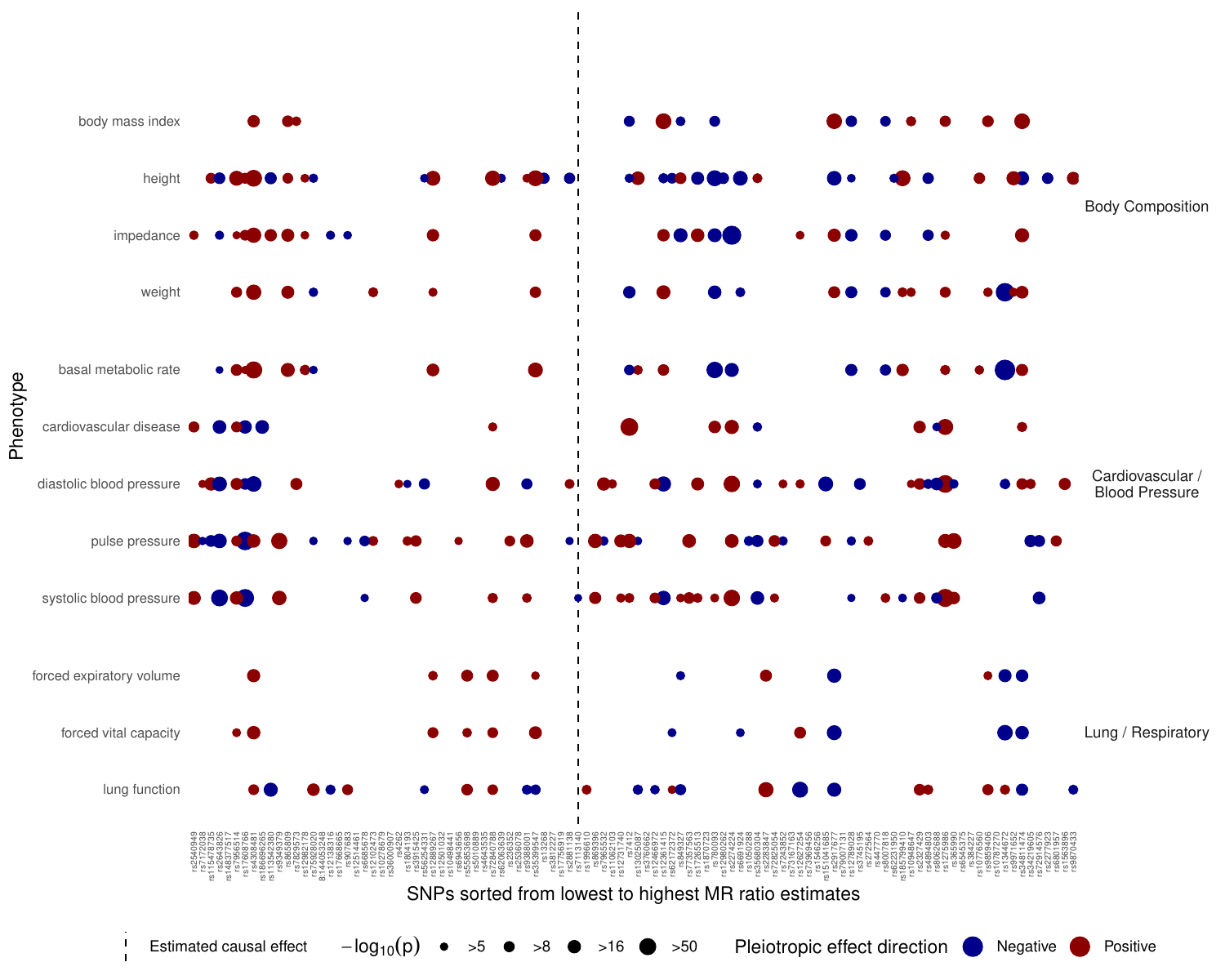}
\caption{\textbf{Exploratory analysis to identify potential pleiotropic pathways biasing MR estimates of the causal effect between \ac{RHR} and \ac{AF}.} Phenome-wide association study (PheWAS) using the OpenGWAS database~\citep{Elsworth2020}, with most significant associations ($p < 10^{-5}$) between SNPs and traits reported for various phenotypic categories. The SNPs on the x-axis are sorted from lowest to highest MR ratio estimates, and the dashed line represents the relative position of the estimated causal effect between \ac{RHR} and \ac{AF} from MR-Quantile.}\label{fig:heatmap}
\end{figure}

\clearpage
\subsection*{Supplementary Tables}

\setcounter{table}{0}
\renewcommand{\thetable}{S\arabic{table}}

\begin{table}[h]
\centering
\caption{\textbf{Computation time (seconds) for compared MR methods:} Values are reported as mean (standard deviation) over 100 datasets for the third simulation scenario with $\theta_0=0.1$ and $q=60\%$ invalid IVs.}
\label{tab:runtime}
\begin{tabular}{lll}
\\
\toprule
 & \multicolumn{2}{c}{Number of instruments} \\
\cmidrule(lr){2-3} Method
       & \multicolumn{1}{l}{$p = 30$} & \multicolumn{1}{l}{$p = 100$} \\
\midrule
MR-Egger               & 0.003 (0.001) & 0.003 (0.001) \\
MR-IVW                 & 0.002 (0.004) & 0.003 (0.003) \\
MR-Lasso               & 0.02 (0.003)  & 0.02 (0.001)  \\
MR-Weighted-Median     & 0.03 (0.004)  & 0.03 (0.001)  \\
MR-RAPS                & 0.10 (0.08)   & 0.11 (0.02)   \\
cML-MA-BIC             & 0.04 (0.005)  & 0.18 (0.02)   \\
MR-Quantile            & 0.26 (0.03)   & 0.30 (0.03)   \\
MR-Weighted-Mode       & 0.94 (0.06)   & 0.98 (0.05)   \\
MR-ContMix             & 0.03 (0.12)   & 1.52 (13.97)  \\
MR-Mix                 & 2.05 (0.06)   & 3.28 (0.19)   \\
cML-MA-BIC-DP          & 8.17 (0.95)   & 41.05 (8.76)  \\
MR-PRESSO              & 15.26 (0.42)  & 68.90 (13.28) \\
\bottomrule
\end{tabular}
\end{table}

\begin{table}[ht]
\small
\centering
\caption{Publicly available datasets from the OpenGWAS database used in the case study}
\label{tab:opengwas}
\begin{tabular}{llllllc}
  \hline
Category & Trait & OpenGWAS ID & Year & Author & PMID & \makecell{Number of significantly \\ associated IVs} \\ 
  \hline \\
  Body Composition & body mass index & ebi-a-GCST90013974 & 2021 & Mbatchou J & 34017140 &   1 \\ 
   & body mass index & ebi-a-GCST90018947 & 2021 & Sakaue S & 34594039 &   1 \\ 
   & body mass index & ebi-a-GCST90029007 & 2018 & Loh PR & 29892013 &   6 \\ 
   & body mass index & ieu-b-40 & 2018 & Yengo, L & 30124842 &   6 \\ 
   & height & ebi-a-GCST90018959 & 2021 & Sakaue S & 34594039 &   3 \\ 
   & height & ebi-a-GCST90029008 & 2018 & Loh PR & 29892013 &  31 \\ 
   & height & ukb-b-16881 & 2018 & Ben Elsworth &  &   4 \\ 
   & impedance & ukb-b-19379 & 2018 & Ben Elsworth &  &   4 \\ 
   & impedance & ukb-b-19921 & 2018 & Ben Elsworth &  &   2 \\ 
   & impedance & ukb-b-7376 & 2018 & Ben Elsworth &  &  17 \\ 
   & impedance & ukb-b-7859 & 2018 & Ben Elsworth &  &   1 \\ 
   & weight & ebi-a-GCST90018949 & 2021 & Sakaue S & 34594039 &   5 \\ 
   & weight & ukb-b-11842 & 2018 & Ben Elsworth &  &  16 \\ 
  \\
  Cardiovascular /  & basal metabolic rate & ebi-a-GCST90029025 & 2018 & Loh PR & 29892013 &  21 \\ 
   \ Blood Pressure & cardiovascular disease & ebi-a-GCST90029019 & 2018 & Loh PR & 29892013 &  14 \\ 
   & diastolic blood pressure & ebi-a-GCST90000059 & 2020 & Surendran P & 33230300 &   2 \\ 
   & diastolic blood pressure & ebi-a-GCST90000063 & 2020 & Surendran P & 33230300 &   2 \\ 
   & diastolic blood pressure & ebi-a-GCST90014017 & 2021 & Mbatchou J & 34017140 &   3 \\ 
   & diastolic blood pressure & ebi-a-GCST90018952 & 2021 & Sakaue S & 34594039 &   1 \\ 
   & diastolic blood pressure & ebi-a-GCST90025981 & 2021 & Barton AR & 34226706 &   1 \\ 
   & diastolic blood pressure & ebi-a-GCST90029010 & 2018 & Loh PR & 29892013 &   3 \\ 
   & diastolic blood pressure & ieu-b-39 & 2018 & Evangelou, E & 30224653 &  21 \\ 
   & diastolic blood pressure & ukb-b-7992 & 2018 & Ben Elsworth &  &   1 \\ 
   & pulse pressure & ebi-a-GCST90000061 & 2020 & Surendran P & 33230300 &   4 \\ 
   & pulse pressure & ebi-a-GCST90000065 & 2020 & Surendran P & 33230300 &   2 \\ 
   & pulse pressure & ebi-a-GCST90018970 & 2021 & Sakaue S & 34594039 &  27 \\ 
   & pulse pressure & ieu-b-5140 & 2024 & Tang H &  &   4 \\ 
   & systolic blood pressure & ebi-a-GCST90000062 & 2020 & Surendran P & 33230300 &   4 \\ 
   & systolic blood pressure & ebi-a-GCST90000066 & 2020 & Surendran P & 33230300 &   3 \\ 
   & systolic blood pressure & ebi-a-GCST90014018 & 2021 & Mbatchou J & 34017140 &   1 \\ 
   & systolic blood pressure & ebi-a-GCST90025968 & 2021 & Barton AR & 34226706 &   1 \\ 
   & systolic blood pressure & ebi-a-GCST90029011 & 2018 & Loh PR & 29892013 &   1 \\ 
   & systolic blood pressure & ieu-b-38 & 2018 & Evangelou, E & 30224653 &  19 \\ 
   & systolic blood pressure & ukb-b-20175 & 2018 & Ben Elsworth &  &   1 \\ 
  \\
  Lung /  & forced expiratory volume & ukb-b-13405 & 2018 & Ben Elsworth &  &   1 \\ 
   \ Respiratory & forced expiratory volume & ukb-b-19657 & 2018 & Ben Elsworth &  &  10 \\ 
   & forced vital capacity & ieu-b-106 & 2021 & Higbee D &  &   1 \\ 
   & forced vital capacity & ukb-b-14713 & 2018 & Ben Elsworth &  &   2 \\ 
   & forced vital capacity & ukb-b-7953 & 2018 & Ben Elsworth &  &   9 \\ 
   & lung function & ebi-a-GCST90029026 & 2018 & Loh PR & 29892013 &  18 \\ 
   & lung function & ebi-a-GCST90029027 & 2018 & Loh PR & 29892013 &   6 \\ \\
   \hline
\end{tabular}
\end{table}

\end{document}